\documentclass[aps,accepted=2024-06-26,a4paper, longbibliography, nofootinbib]{quantumarticle}
\pdfoutput=1

\usepackage[utf8]{inputenc}
\usepackage{amssymb}
\usepackage{babel}
\usepackage{amsmath}
\usepackage{amsthm}
\usepackage{braket}
\usepackage{accents}
\usepackage[most]{tcolorbox}
\usepackage{microtype}
\usepackage{natbib}
\usepackage{bbm}
\usepackage{color}
\usepackage[colorlinks = true, urlcolor = blue]{hyperref}
\usepackage{xcolor}
\usepackage[ruled,vlined,linesnumbered]{algorithm2e}
\usepackage{algcompatible}
\usepackage{ifthen}
\usepackage{caption}
\usepackage[format=plain]{subcaption}
\newcommand{\revision}[1]{#1}

\newcommand{\pc}{\mathcal{P}}
\newcommand{\bc}{\mathcal{B}}
\newcommand{\gc}{\mathcal{G}}

\newcommand{\scc}{\mathcal{S}}
\newcommand{\tsc}{\tilde{\mathcal{S}}}

\newcommand{\tby}{\tilde{\mathbf{y}}}
\newcommand{\ty}{\tilde{y}}
\newcommand{\by}{\mathbf{y}}
\newcommand{\Bm}{\mathbf{m}}

\newcommand{\bphi}{\mathbf{\Phi}}
\newcommand{\lc}{\mathcal{L}}

\newcommand{\uc}{\mathcal{U}}

\newcommand{\tr}{\text{Tr}}
\newcommand{\kp}{k^\prime}
\newcommand{\lp}{l^\prime}
\newcommand{\Ap}{A^\prime}

\newcommand{\Kket}[1]{\ket{#1}\rangle}
\newcommand{\Bbra}[1]{\langle \bra{#1}}
\newcommand{\BK}[1]{\langle \braket{#1} \rangle}

\newcommand{\mnorm}[1]{|| #1 ||_{\textrm{mix}}}
\newcommand{\unorm}[1]{|| #1 ||_{\textrm{uni}}}

\newtheorem{theorem}{Theorem}
\newtheorem{corollary}{Corollary}
\newtheorem{lemma}{Lemma}

\newtheorem{definition}{Definition}
\newtheorem{condition}{Condition}
\newtheorem{fact}{Fact}
\long\def\ca#1\cb{} %Use for commenting out: \ca...\cb

\allowdisplaybreaks

\begin{document}

\title{Universal framework for simultaneous tomography of quantum states and SPAM noise}

\author{Abhijith Jayakumar}
\email{abhijithj@lanl.gov}
\affiliation{Theoretical Division, Los Alamos National Laboratory, 87545, NM, USA}
\author{Stefano Chessa}
\email{schessa@illinois.edu}
\affiliation{Theoretical Division, Los Alamos National Laboratory, 87545, NM, USA}
\affiliation{NEST, Scuola Normale Superiore and Istituto Nanoscienze-CNR, I-56126, Pisa, Italy}
\affiliation{Electrical and Computer Engineering, University of Illinois Urbana-Champaign, Urbana, 61801, IL, USA}
\author{Carleton Coffrin}
%\email{}
\affiliation{Los Alamos National Laboratory, Los Alamos, 87545, NM, USA}
\author{Andrey Y. Lokhov}
% \email{lokhov@lanl.gov}
\affiliation{Theoretical Division, Los Alamos National Laboratory, 87545, NM, USA}
\author{Marc Vuffray}
%\email{}
\affiliation{Theoretical Division, Los Alamos National Laboratory, 87545, NM, USA}
\author{Sidhant Misra}
\email{sidhant@lanl.gov}
\affiliation{Theoretical Division, Los Alamos National Laboratory, 87545, NM, USA}

%\date{\today}

\begin{abstract}
We present a general denoising algorithm for performing \emph{simultaneous tomography} of quantum states and measurement noise. This algorithm allows us to fully characterize state preparation and measurement (SPAM) errors present in any quantum system. Our method is based on the analysis of the properties of the linear operator space induced by unitary operations. Given any quantum system with a noisy measurement apparatus, our method can output the quantum state and the noise matrix of the detector up to a single gauge degree of freedom. We show that this gauge freedom is unavoidable in the general case, but this degeneracy can be generally broken using prior knowledge on the state or noise properties, thus fixing the gauge for several types of state-noise combinations with no assumptions about noise strength. Such combinations include pure quantum states with arbitrarily correlated errors, and arbitrary states with block independent errors. This framework can further use available prior information about the setting to systematically reduce the number of observations and measurements required for state and noise detection. Our method effectively generalizes existing approaches to the problem, and includes as special cases common settings considered in the literature requiring an uncorrelated or invertible noise matrix, or specific probe states.

\end{abstract}

\maketitle

\section{Introduction}
Quantum computing promises to have the potential to solve complex problems that are beyond the reach of classical computers \cite{Harrow2017, aaronson_chen2016, PhysRevA.88.022316, Yung2018}, but realizing this full potential requires overcoming the various challenges posed by noise \cite{Preskill_2018, 9719705, chen2022}. %Because of the fundamental laws of quantum mechanics errors affect quantum devices \cite{PhysRevA.51.992, RevModPhys.82.1155, gardiner2004quantum} and that must be overcome in order to perform accurately the tasks that such devices are built for. 
These errors can arise from a number of sources, including noise in the specific hardware architectures \cite{Ladd2010, Buluta_2011, Saffman_2016, Bruzewicz2019, Slussarenko2019, Henriet_2020, Kinos2021, Chaurasiya2022}, inaccuracies in control and limitations in the operations that can actually be performed on such systems \cite{Almudever2017, Vandersypen2017, Reilly2019, Corcoles2020, Leon2021}.  

To tame these errors, researchers have developed various and still growing number of approaches and strategies that can be included in the macro-categories of quantum error correction \cite{Shor1996, Gottesman1997, Cory1998, Preskill1998, Knill2000, Chiaverini2004, lidar_brun_2013, Campbell2017, Roffe2019}, quantum error mitigation \cite{Temme2017, Endo2018, Kandala2019, Takagi2022, Cai2022}, and noise learning, which includes specific techniques such as, among others, quantum process tomography \cite{Poyatos1997, Chuang1997, Mohseni2008, Merkel_QProcTom_2013}, gate set tomography \cite{Greenbaum2015, Nielsen_2021} and randomized benchmarking \cite{Knill2008, Magesan2011, Magesan2012, Helsen2022}. 

Among these sources of noise, state preparation and measurement (SPAM) errors can prove to be particularly significant. As an example for the current best superconducting qubit-based devices, they can be in the range 1-3\%, see e.g. \cite{Arute2019, Elder2020, Opremcak2021}. These errors occur when the initial state of a quantum system and/or the measurement of its final state are not precisely known or controlled. SPAM errors can result in systematic biases that can greatly impact the accuracy of quantum information processing in noisy devices both in quantum error correction and in the so-called  ``noisy intermediate scale quantum'' (NISQ) tasks see e.g. \cite{Arute2019, Wright2019, Ryan2021}. SPAM errors are the focus of this paper, specifically, we address the issue of the simultaneous correct identification of the (possibly arbitrarily correlated and of arbitrary strength) noise affecting detectors after the preparation of a state $\rho$ \textit{and} the correct identification of $\rho$ itself. Despite these two tasks being some of the most fundamental operations one could imagine for quantum information processing, their simultaneous realization is hindered by the fact that state preparation and measurement noise matrix can be determined only up to a gauge transformation \cite{Gauge1, Gauge2, Blume2013}. This fact presents a severe limitation  for state tomography and noise characterization as the knowledge of the real underlying noise process is essential for diagnostics and the optimization of the device. To address this in recent years some attempts have been made to develop techniques that resolve these kinds of gauge degeneracy \cite{ Di_Matteo_2020, Lin_SPAM_2021, Laflamme2022, lin2019freedom}. 

Our work presents a significant contribution in this direction: we provide a general framework for identifying conditions under which noise models and prepared states can break the gauge freedom, which includes as special cases many previously proposed approaches. We achieve this by introducing a denoising algorithm that can simultaneously estimate both the state of the system and detector noise up to a single gauge parameter. The output of this algorithm gives a complete characterization of the SPAM errors in the system as it gives the maximum possible information about the true state prepared in the system and the stochastic matrix governing the measurement noise. 

The main contributions of this work are as follows: First, we completely characterize the gauge freedom in our problem, and prove that the simultaneous characterization of state and noise is only hindered by a single gauge parameter. Next, we give a general algorithm to simultaneously estimate a quantum state and any stochastic matrix characterizing SPAM errors in a quantum system, up to this unavoidable gauge parameter. We also outline methods using which this gauge can be fixed given many forms of prior information about the state or the noise matrix, including practically relevant cases, such as states with known purity, independent ancilla qubits, and known expectation values. To address more practical settings, we devise a randomized version of our algorithm that uses computational basis measurements that only involves the application of Clifford circuits. Finally, we also provide a sample complexity analysis of our algorithm and show that the number of samples required depends naturally on the distance of the state and the noise matrix from a maximally mixed case. 
The paper is structured as follows:

\begin{itemize}
    \item In Sec.~\ref{sec:problem_statement}: we state the problem of noise-state \textit{simultaneous tomography}, set the notation, and discuss the gauge freedom intrinsic in the problem. Here, we prove that the problem has only a single gauge degree of freedom.
    \item In Sec.~\ref{sec: Simultaneous tomography}: we outline the noise-state simultaneous tomography algorithm for any POVM. We also show the special case of the algorithm using computational basis measurements and derive the sample complexity of the randomized version. We support our analysis of the randomized algorithm using numerical results that show the tightness of our analysis.
    \item In Sec.~\ref{sec:fixing}: we show how prior knowledge about the system can be used to fix the gauge and also to improve the algorithm in terms of resource efficiency.
    \item In Sec.~\ref{sec: discussion}: we draw the conclusions and discuss the perspectives of this work.
\end{itemize}
The summary of this structure and our approach is provided in Fig.~\ref{fig:summary}.

\begin{figure*}[!htb]
    \includegraphics[width=\textwidth]{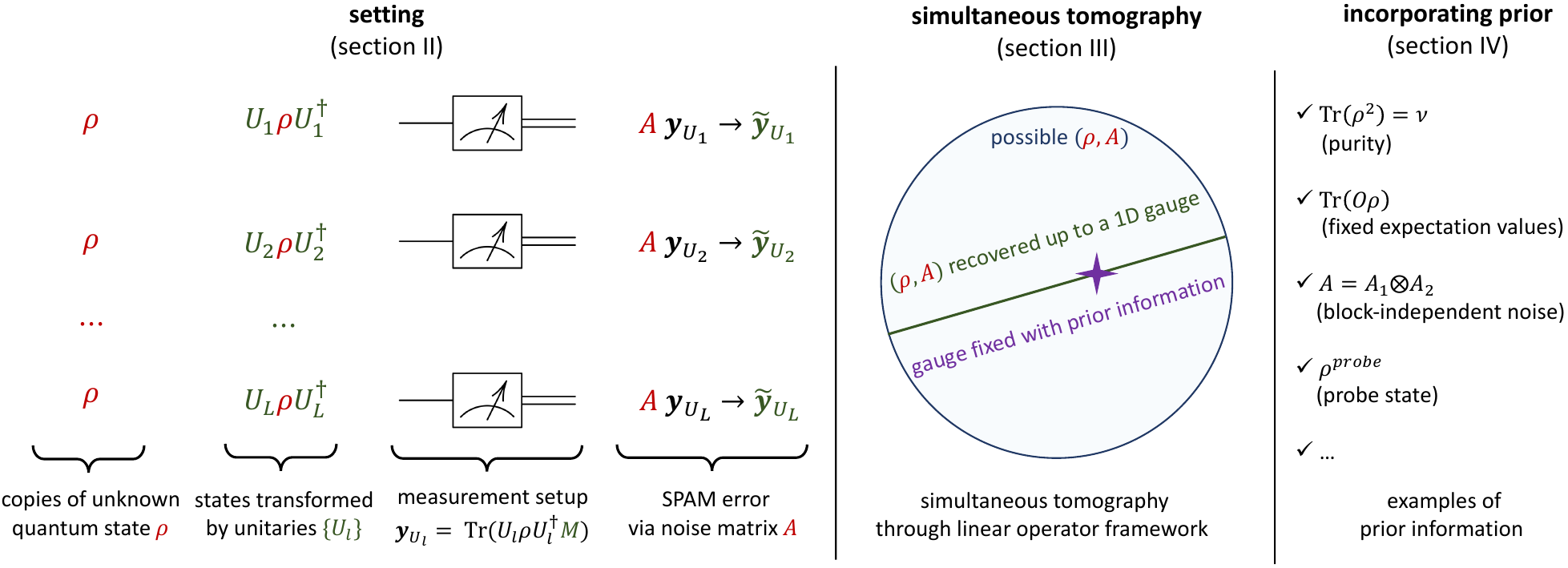}
    \caption{{\bf Summary of our approach to simultaneous tomography of quantum states and SPAM noise.} In section II of the paper, we state the problem of \emph{simultaneous tomography}: the process of estimating the quantum state ($\rho$) and the noise matrix associated with measurement errors ($A$) using a unified set of measurements. In section III, we introduce a universal algorithm for performing the simultaneous tomography in full generality up to the fundamental and unavoidable gauge ambiguity, prove that this degeneracy is one-dimensional, and discuss the sample-complexity of our algorithm. Finally, in section IV, we provide many examples of prior information about either the state or the noise that allows one to unambiguously recover the quantum state and the noise matrix. These examples include many settings considered in previous work.}
    \label{fig:summary}
\end{figure*}

\section{Problem statement, setting, and notation}  \label{sec:problem_statement}

\subsection{The problem of simultaneous tomography}
We consider the problem of fully characterizing persistent errors as well as recovering the underlying quantum state in a quantum system affected by imperfect state preparation protocols or measurement errors.
The quantum system consisting of $n$ qubits is prepared in a state $\rho$ and measured using a general \emph{Positive Operator Valued Measure} (POVM) \cite{nielsen2002quantum}. Given the $n-$qubit POVM, $\{M_k| \quad k \in [D],~ M_k  \succcurlyeq 0, ~\sum_k M_k = I \},$ we define the measurement probabilities obtained after applying a unitary transformation ($U$) to the quantum state  as follows,
\begin{align}
\label{eq:yu_def}
    y_k(U) :=  \tr (U \rho U ^\dagger  M_k), \quad k=1,\ldots,D.
\end{align}
Now we model the measurement noise in the quantum system as a general stochastic matrix, $A$, acting on the probability distributions defined in \eqref{eq:yu_def}.
\begin{equation}
\tilde{y}_k(U) :=  \sum_{k^\prime \in [D]} A_{k\, k^\prime}~ y_{k^\prime}(U).
\end{equation}
We will simplify the notation to  $y_k$ and $\tilde{y}_k$ in the case where $U$  is just the identity transformation. This transition matrix ($A$) model is quite universal as \revision{ it models a general measurement error one can have assuming that this error is independent of the operations performed on the computer before measurement. This is a commonly assumed simplification in the SPAM literature \cite{Lin_SPAM_2021, Laflamme2022, Gauge1}. We also assume that any unitary $U$ in the above equations can be applied without any errors.}

The main aim of this work will be to use these types of noisy measurements to fully characterize the system and the noise. We refer to this task as  \emph{simultaneous tomography}, which can be defined as the problem of designing a set of unitary  operators $U_1, \ldots, U_l$ that are efficiently implementable on a given quantum system, and a procedure that uses the noisy measurements $\tilde{y}_k(U_1), \ldots, \tilde{y}_k(U_l)$ along with prior information about the system  to estimate the state $\rho$ and the matrix of measurement noise $A$. 

Simultaneous tomography is directly related to SPAM error characterization. The recovered state $\rho$ can be compared with the state that was intended to be prepared and the state preparation error rates can be computed from their difference \cite{Lin_SPAM_2021}. While the noise matrix $A$ represents the errors in measurement device.\\

The number of gates, measurements, and classical processing required to perform simultaneous tomography is expected to be larger than those for noiseless state tomography. The exact overhead depends on the structure of the noise, the underlying state, and access to prior knowledge of the state and noise model. In practice, the best choice of gates $U_1, \ldots U_L$ will depend on which gates are native and least noisy for the specific quantum system in consideration. 

\subsection{Linear operator framework} \label{sec:linear_operator}
The process of simultaneous tomography consists of two steps: (i) implementing a set of chosen unitaries on the quantum system and obtaining the corresponding noisy measurements, and (ii) performing a set of classical post-processing computations on the measurements to obtain the estimates of the state and measurement noise. By considering only linear classical post-processing, the overall procedure can be viewed as a linear transformation on the underlying state which we describe below. 

In a quantum system, the action of any unitary on a state ($\rho \rightarrow U \rho U^\dagger$) can be represented by a linear \textit{superoperator}. To demarcate between operators and superoperators, we use the standard notation of $\Kket{\rho}$ for the $4^n$-dimensional vector representing $\rho$ in the space acted on by superoperators \cite{Nielsen_2021}. Other objects, such as the $2^n \times 2^n$ identity matrix $I$ or POVM operators, can be represented by a $4^n$-dimensional vector in a similar way. Naturally, \revision{for two operators $P$ and $Q$ the inner-product  $\BK{P|Q}$ is defined as $\BK{P|Q} = \tr{(P^\dagger Q)}$}. 

\revision{It is advantageous to isolate the action of a unitary on the traceless subspace of operators},
\begin{equation}
    U \rho U^\dagger  = \Phi(U) \Kket{\rho}  = \frac{\Kket{\hat{I}}}{2^{n/2}} + \phi(U) \Kket{\bar{\rho}} \; .
\end{equation}
Here, $\Phi(U)$ is the complete superoperator corresponding to the action of $U$, $\phi(U)$ is the traceless part of this superoperator. We also define $\hat{I} = I /2^{n/2}$,  the normalized identity matrix, as well as $\Kket{\bar{\rho}} = \Kket{\rho} - \Kket{\hat{I}}/2^{n/2}$ for the $4^n$-dimensional vector representing the traceless part of $\rho$.

In this notation, the noisy measurements take the form, 
\begin{equation}\label{eq:basis_free}
\tilde{y}_k(U) =  \sum_{k^\prime \in [D]} A_{k\,k^\prime}  \left( \frac{\BK{M_{\kp}|\hat{I}}}{2^{n/2}} + \Bbra{M_{\kp}}\phi(U) \Kket{\bar{\rho}} \right).
\end{equation}
In general, this expression can be expanded using any basis in the traceless subspace. Let $\mathcal{B}_L$ and $\mathcal{B}_R$ be two sets of traceless and hermitian operators. \revision{Further, assume that $\mathcal{B}_R$ is a normalized set of operators (i.e. $\BK{P|P} = 1$) whose linear span is the space of all traceless hermitian operators. Also, assume that the measurement operators lie in the linear span of $\mathcal{B}_L$}. Then we can always expand the state in one of the basis sets, and the measurement operators in the other as follows:
\begin{align}
    \Kket{\bar{\rho}} &= \sum_{P \in \bc_R} s_P \Kket{P} \; , \\ 
    m_{k\, I} &= \BK{M_{k}|\hat{I}}, \quad m_{k\,Q} = \BK{M_{k}|Q} \; , \ Q \in \bc_L \; .
\end{align}
Using these relations we can expand \eqref{eq:basis_free} in this specific basis as, 
\begin{align} %  \label{eq:y_linear_operator}
\ty_k(U) = \sum_{\kp  \in [D]} &\dfrac{A_{k\, \kp } m_{\kp I}}{2^{n/2}} \nonumber ~~+ \\ &\sum_{\substack{\kp \in [D], \\ P \in \bc_R, Q \in \bc_L} } \hspace{-0.3cm} s_P \, \phi(U)_{P\,Q} \, A_{k\, \kp}   m_{\kp\, Q} \;.   \label{eq:noisy_observations}
\end{align}  
\vspace{-0.05in}
 Here $\phi(U)\Kket{P} =  \sum_{Q \in \mathcal{B}_L} \phi(U)_{P\,Q} \Kket{Q}  + \Kket{b^\perp_L}$, where $\Kket{b^\perp_L}$ is an operator orthogonal to every operator in $\mathcal{B}_L$.

As an example of the setup described above, take the set of all $n$-qubit normalized Pauli strings, $\hat{\pc} \equiv \{\frac{I}{\sqrt2}, \frac{X}{\sqrt2},\frac{Y}{\sqrt2},\frac{Z}{ \sqrt2} \}^{\otimes n}$. We can take the basis sets to be the traceless operators in this set, $\bc_L = \bc_R = \hat{\pc} \setminus \{ \hat{I} \}.$ In this case, \revision{ the matrix $\phi(U) \in \mathbb{R}^{4^n-1 \times 4^n-1}$} will just be the well-known Pauli Transfer Matrix representation for $U.$ And $s_P$ would simply be expectation values of the state with the  Pauli strings \cite{Lin_SPAM_2021}.

\begin{centering}
\begin{tcolorbox}[enhanced jigsaw,
                  colback=gray!6,
                  colframe=black,
                  width=8.5cm,
                  arc=1mm, auto outer arc,
                  boxrule=0.5pt,
                 ]
{\bf Running example:}
To illustrate the ideas in this paper we will use an example of a noisy two-qubit system. The same system will be used throughout the paper at various points as a pedagogical tool.
\\

Consider a $2$-qubit system with $\rho = \ket{01}\!\!\bra{01} = \frac{(I+Z)\otimes(I-Z)}{4}$. If $\bc_R$ is the normalized Pauli basis, then the non-zero coefficients are $s_{I\otimes I},s_{Z\otimes I} = \frac12, ~~s_{I\otimes Z}, s_{Z\otimes Z} = -\frac{1}{2}.$ For measurements in the computational basis the noise matrix is a $4\times 4$ matrix which we take to be $A = (0.9I + 0.1 X)^{\otimes 2}.$
\end{tcolorbox}
\end{centering}

To perform simultaneous tomography, the noisy measurements $\ty_k(U)$ are passed through \emph{linear classical post-processing} where we compute linear combinations 
\begin{align}   \label{eq:defining_z}
    z_k = \sum_{l} c_l \ty_k(U_l) \; , \quad \mbox{where} \ \sum_l c_l = 1 \; . 
\end{align}
Using \eqref{eq:basis_free} the quantities $z_k$  \revision{can be expressed in a basis independent fashion as,}
\begin{align}
\label{eq:z_def}
z_k = \hspace{-0.06in} \sum_{\kp \in [D]}& \frac{A_{k \, \kp} m_{\kp \, I}}{2^n}  \nonumber \\ &+\hspace{-0.06in} \sum_{\kp \in [D]} \hspace{-0.06in} A_{k \, \kp}\BK{ M_{\kp}|\left(\sum_l c_l \phi(U_l)\right)| \bar{\rho}}.
\end{align}
Thus computing the quantities $z_k$ can be viewed as applying the \emph{effective} non-unitary linear transformation 
\begin{align}   \label{eq:linear_operators}
    \Phi = \sum_l c_l \Phi(U_l)
\end{align}
on the state $\rho$ and then obtaining noisy measurements. The affine constraint ($\sum_l c_l =1$) makes these transformations trace-preserving.

In the context of simultaneous tomography, the following two points are important regarding these linear operators. First, which of these linear transformations are sufficient for successfully performing simultaneous tomography? Second, what set of unitaries is required for efficiently implementing the linear transformations \eqref{eq:z_def}?

To this end, let $\mathcal{U}(2^n)$ be the unitary group on $n$ qubits. For a chosen subset, $\mathcal{S} = \{U_1, \ldots, U_l\} \subseteq \uc(2^n)$, the overall computational power of implementing them on the quantum system and performing classical linear post-processing can be summarized by the linear operator space defined by. 
\begin{align}   \label{eq:linear_operator_space}
    \lc(\mathcal{S}) = \left\{\sum_{l}c_l \Phi(U_l) \, \middle\vert \, U_l \in \mathcal{S},  \ \sum_{l} c_l = 1 \right\}.
\end{align}

Using the linear operator space allows us to view the requirements of a given simultaneous tomography task in a general way without reference to a chosen basis or a given set of gates. If it is determined that a such subset of unitaries is sufficient for a given simultaneous tomography task, then depending on the quantum system there may be multiple ways of realizing this subset.  \revision{Note that $\mathcal{L}$ does not correspond to any single quantum channel, rather it represents a set of quantum channels.}

We call a set of unitaries $S_g \subseteq \uc(2^n)$ the \emph{generator set} for a given subset $S \subseteq \uc(2^n)$ \revision{if $S_g$ has fewer elements than $S$}, and $\lc(S) \subseteq \lc(S_g)$. Notice that there might be multiple ways to choose $S_g$ given $S$. The appropriate choice  will depend for example on what set of gates are least noisy and natively available on a given quantum architecture and how much classical post-processing power is available. Working with the super-operator space $\lc(S)$ allows us to separate \emph{what} is needed to perform simultaneous tomography and \emph{how} to realize it with a given quantum device and classical processing resource.

For instance,  we will show that using a complete set of superoperators $\lc(\uc(2^n))$ is sufficient to perform simultaneous tomography. But for performing simultaneous tomography in the computational basis,  we can also use a smaller subset of the Clifford group as a generator set for this super-operator space.\revision{Using only $\text{CNOT}$ and arbitrary single qubit gates, this generator set can be implemented with circuits of linear depth \cite{maslov2018shorter}}. We will also discuss some cases of performing this task in the presence of prior information where a limited subset $\lc(S) \subset \lc(\uc(2^n))$ is sufficient.

\subsection{Identifiability for the simultaneous tomography problem}
Given noisy observations of the form in \eqref{eq:noisy_observations}, the pertinent question is whether simultaneous tomography of $\rho$ and $A$ is even possible without any additional information? Interestingly, the answer is ``no" in the most general case. To see this, consider a one-parameter family of transformations on the state and noise channel defined as follows,
\begin{multline}
\label{eq:gauge}
 \revision{  A_{k \,\kp} \rightarrow A^\prime_{k\,\kp}(\alpha) =  \alpha A_{k\,\kp} + (1 - \alpha) \frac{\sum_{j \in [D]}A_{k\, j} m_{j\,I}}{2^{n/2}} \;} , \\ 
   \rho \rightarrow \rho^\prime(\alpha) =   \frac{\rho}{\alpha} + \left(1 - \frac{1}{\alpha} \right) \frac{I}{2^n} \; ,~~\alpha \in \mathbb{R}\setminus\{0\} \; .
\end{multline}

By simple algebra, we can check that this simultaneous transformation of the state and noise will leave the noisy outputs in \eqref{eq:noisy_observations} invariant thus leaving us with no means to distinguish between them. In literature, this kind of invariance has been called \emph{gauge freedom} \cite{Gauge1, Gauge2, Cai2022}. The gauge freedom implies that any simultaneous tomography method will have at least a one-parameter ambiguity. \revision{These gauge transformations represent a one-parameter manifold in the $(\rho,A)$ space. While the transformations are mathematically well defined for any non-zero $\alpha$, the set of physically allowed $\alpha$ will be those such that $\rho'(\alpha), A'(\alpha)$ are respectively valid density and stochastic matrices. But even these physical constraints cannot unambiguously fix $\alpha$ in general. }

\revision{The gauge freedom can be viewed as the inability to separate whether the randomness in the observations comes from the random nature of quantum measurements or if it is a product of classical noise. An extreme example is as follows; suppose we are given a single qubit state and a noisy measurement apparatus. Suppose we also observe that when this qubit is measured in the computational basis after applying any $U$, both $1$ and $0$ are seen with equal probability. Given such a system there is no way to distinguish whether the state is maximally mixed or whether it is the measurement device that is completely noisy. However, if we have prior information (confidence about the state preparation itself) that the state is pure, we can ascertain that the randomness came from the measurement device.  The gauge freedom in the simultaneous tomography problem generalizes this inherent ambiguity in the problem.}

The question remains whether there are other transformations that also leave \eqref{eq:noisy_observations} invariant. The theorem below shows that the transformation in \eqref{eq:gauge} represents the only possible ambiguity in the problem.

\begin{theorem}{\bf Gauge freedom is the only ambiguity.}\\
\label{thm:degeneracy_char}
Let $\tby_{A,\rho}(U)$ be the noisy measurement distribution produced by the quantum state $U\rho U^\dagger,$ with the noise characterized by $A$, as in \eqref{eq:noisy_observations}. If for another system in a state $\rho^\prime$ with noisy measurements characterized by $A^\prime,$ it is given that $\tby_{A,\rho}(U) = \tby_{A^\prime,\rho^\prime}(U),~ \forall~ U \in \mathcal{U}(2^n)$,  then there must exist a gauge parameter $\alpha \in \mathbb{R} \setminus \{ 0 \}$ such that \eqref{eq:gauge} holds.
\end{theorem}

The proof of this theorem rests on the fact that when we have access to all possible unitary gates in $\uc(2^n)$, the induced linear operator space $\lc(\uc(2^n))$ defined in \eqref{eq:linear_operator_space} is \emph{complete}. The precise statement of this completeness result is given in Appendix \ref{app:PTM}. The full proof of the theorem can be found in Appendix \ref{app:gauge}.

This gauge ambiguity can be overcome if we have some prior information about the system that uniquely identifies the correct $\rho$ and $A$ from the one-parameter family in \eqref{eq:gauge}. In Sec. \ref{sec:fixing}, we show multiple, physically and operationally relevant cases of prior information that can fix this gauge.

\section{Simultaneous tomography: conditions and algorithm}\label{sec: Simultaneous tomography}

In this section, we will demonstrate how noisy measurements generated according to \eqref{eq:noisy_observations}, can be used to reconstruct both $\rho$ and $A$ up to a single gauge parameter. As in the case of noiseless tomography, simultaneous tomography can also be performed by using measurement outcomes produced by observing the state after rotating it using a set of pre-defined unitary operations.

\subsection{Sufficient conditions for simultaneous tomography}
Beyond the gauge degree of freedom, few edge cases can make simultaneous tomography impossible.  For instance, if the $A$ matrix always outputs the uniform distribution in  $D$ dimensions, then we can never recover the exact state $\rho$ from the noisy measurement outcomes. To avoid these types of pathological cases we assume that the output of $A$ always has some correlation with the input:
\begin{condition}
\label{cond1}
$\exists k,i,j \in [D], ~\text{such that}~ A_{ki} \neq A_{kj}.$
\end{condition}
% \begin{equation}\label{eq:cond1}
%    \textbf{(Condition 1)}~~\exists k,i,j \in [D], ~\text{such that},~ A_{ki} \neq A_{kj} \; .
% \end{equation}
If this condition does not hold, then the probability of observing a certain output conditioned on an input, $Pr(k|\kp) = A_{k \, \kp},$ would be independent of the input. We call such an $A$ the \emph{erasure channel}.

Similarly, simultaneous tomography is impossible if $\rho$ is a \emph{maximally mixed state} ($\rho \propto I$). This would imply that $U \rho U^\dagger = \rho$, and full information about $A$ would not be recoverable from noisy measurements defined in \eqref{eq:noisy_observations}. To avoid this case we must assume that the state has some non-zero overlap with the space of traceless operators  i.e. at least one of the $s_P$ coefficients is non-zero
\begin{condition}
\label{cond2}
$\exists P \in \bc_R, ~\text{such that},~s_P \neq 0.$
\end{condition}
% \begin{equation}\label{eq:cond2}
%    \textbf{(Condition 2)}~~ \exists P \in \bc_R, ~\text{such that},~s_P \neq 0.
% \end{equation}
Additionally, we also require the set of measurement operators to be linearly independent:
\begin{condition}
$\{M_i|i \in  [D] \}$ are linearly independent.
\label{cond3}
\end{condition}
% \begin{equation}\label{eq:cond3}
%    \textbf{(Condition 3)}~~\{M_i | i \in  [D]  \}~\text{ is linearly independent}.
% \end{equation}
If this condition is not satisfied, then the definition of the measurement operators are non-unique and it is impossible to reconstruct $A$. But given a linearly dependent POVM, we can always construct a reduced set from it such that this new POVM is linearly independent (see Appendix \ref{app:new_POVM}).

To describe our algorithm, in  \ref{subsec:no_shoterror}, we assume that the noisy measurement probabilities are directly available to us, i.e., for any $U$ the distribution $\tby(U)$ is fully specified. We will discuss the more practical variant of our algorithm with finite measurement shots and randomized measurements in \ref{subsec:shoterror}.
\subsection{Simultaneous tomography algorithm}
\label{subsec:no_shoterror}
The algorithm relies on the \emph{completeness} of the linear operator space used in the proof of Theorem~\ref{thm:degeneracy_char}. The proof is a constructive one and naturally leads to the algorithm described in this section. 

While simultaneous tomography can be performed on any basis in the operator space, we find that the presentation of the algorithm simplifies considerably if we fix $\bc_L$ to be the traceless POVM operators,
\begin{equation}
    \bc_L = \{\bar{M}_i| i \in [D] \},
\end{equation}
where, $\bar{M}_i = M_i - \BK{M_i|I} \frac{I}{D}.$

If $\bc_L$ does not span the space of traceless operators, there will be a space orthogonal to it which is unobservable by the POVM. We denote this orthogonal space by $\bc^{\perp}_L.$ As an example, if $\bc_L$ is given by the traceless computational basis measurement operators, then $\bc^{\perp}_L$ will span the space of all off-diagonal operators:
\begin{equation}
    \bc_L^\perp = \{Q|\tr(Q) = 0,~ \BK{Q|Q^\prime} = 0~\forall Q^\prime \in \bc_L,~ Q = Q^\dagger\}.
\end{equation}
Notice that while $\bc_L$ is a basis set, $\bc^\perp_L$ is a vector space.

\begin{centering}
\begin{tcolorbox}[enhanced jigsaw,
                  colback=gray!6,
                  colframe=black,
                  width=8.5cm,
                  arc=1mm, auto outer arc,
                  boxrule=0.5pt,
                 ]
{\bf  Running example:}
For the two qubit system measured in the computational basis $\bc_L = \{\ket{00}\!\!\bra{00} - \frac{I}{4}, \; \ket{01}\!\!\bra{01} - \frac{I}{4},\; \ket{10}\!\!\bra{10} - \frac{I}{4}, \; \ket{11}\!\!\bra{11} - \frac{I}{4}\}.$ Since this spans all tracelss  diagonal operators, $\bc^{\perp}_L$ is the set of  all off-diagonal $2$-qubit operators.
\end{tcolorbox}
\end{centering}

A key step in the algorithm is the construction of a set of  canonical super-operators. The first one is $E_{I}$, which is a trace-preserving superoperator that effectively eliminates all operators in $\bc_R$
\begin{align}
\label{eq:EI}
&E_{I} \Kket{I} = \Kket{I} , ~E_{I} \Kket{P^\prime} \in \bc_L^{\perp}~~~~ \forall  P^\prime \in \bc_R \; . 
\end{align}

Then we define a set of trace-preserving canonical super-operators that effectively maps a specific operator in $\bc_R$ to a specific operator in $\bc_L$. For any $P \in \bc_R$ and $\Bar{M}_i \in \bc_L$
\begin{align}
\label{eq:EPi}
&E_{P,i} \Kket{I} = \Kket{I} ,~E_{P,i} \Kket{P} -  \Kket{\Bar{M}_i} \in \bc_L^\perp , \nonumber \\
&E_{P,i} \Kket{P^\prime} \in \bc_L^{\perp}~~~~ \forall  P^\prime \in \bc_R \setminus \{ P\} \; .
\end{align}

These mappings are ``effective", as they always have some component in $\bc^\perp_L$, which we have left uncharacterized in the above definitions. But this is inconsequential as these components are not observed by the POVM. In what follows, we refer to $E_{I}$ and $E_{P,i}$ as to the \emph{eliminator operators}, or \emph{eliminators}.

\begin{centering}
\begin{tcolorbox}[enhanced jigsaw,
                  colback=gray!6,
                  colframe=black,
                  width=8.5cm,
                  arc=1mm, auto outer arc,
                  boxrule=0.5pt,
                 ]
{\bf  Running example:}
Since the unobservable part is left unspecified, the definition of the eliminators are not unique. For the two-qubit example we take  $\bc_R$ as all normalized, traceless Pauli strings. In that case we can always take,
$E_I = \frac14( \Kket{I}\Bbra{I} + \Kket{X}\Bbra{X})^{\otimes 2}.$

From this definition,
\begin{align}\label{eq:EI_example}
E_I&\Kket{I} = \frac{1}{4} (\tr(I) ~ I + \tr(X \otimes I ) ~X \otimes I \nonumber \\&+ \tr(I \otimes X) ~I \otimes X + \tr(X \otimes X) X \otimes X) = I  
\end{align}
Similarly we can check that this eliminates all the diagonal Pauli strings owing to the anti-commutation relation between $Z$ and $X$. This will not eliminate Pauli strings with only $X$ for instance. But these are off-diagonal operators which lie in $\bc_L^\perp.$ Similarly other eliminators can also be constructed for this case.  The general formula for these constructions in the computational basis is given in Section \ref{subsec:shoterror}.

\end{tcolorbox}
\end{centering}

Now to perform simultaneous tomography, we need  to   apply these canonical operators to the state by aggregating measurement outcomes as described in \eqref{eq:z_def}. To do this, it is sufficient to have a set of unitary operators such that these canonical operators lie in their span. We call such a set of unitaries \emph{tomographically complete} and use noisy measurement outcomes generated by these  unitaries, as in \eqref{eq:noisy_observations}, to perform simultaneous tomography.

\begin{definition}[\bf Tomographically  complete set]
We call a set of unitary operators $\uc_{tom} = \{U_1 \ldots U_L\} $  \textbf{tomographically complete} if for all $P \in \bc_R,~i \in [D]$, we have $E_{P,i} \in  \lc(\uc_{tom})$ and $E_{I} \in \lc(\uc_{tom})$.
\end{definition}

This implies that if the set  $\uc_{tom}$ is tomographically complete, then there exists coefficients $c_l^{P,i}$ such that $\sum_{l} c_l^{P,i} = 1$ and
\begin{align}
    \sum_{l=1}^L c_l^{P,i} \Phi(U_l) = E_{P,i}.
\end{align}
Further, there exist coefficients, $c^{I}_l$, such that $\sum_l c_l^{I} = 1$ and,
\vspace{-0.2in}
\begin{align}
    \sum_{l=1}^L c_l^{I} \Phi(U_l) = E_{I}.
\end{align}
The definition of this set of unitaries and the corresponding coefficients for constructing the eliminators will obviously depend on the basis sets $\bc_R$ and $\bc_L.$  For the special case of computational basis measurements with $\bc_R$ taken to be the Pauli operator basis, we can show that this set is a subset of the Clifford group on $n-$qubits (see \eqref{eq:EPi_def}). 

Now using these coefficients in \eqref{eq:defining_z} we can aggregate the noisy measurement outcomes to effectively apply the canonical operators to the state,
\begin{align}
&z^{I}_k := \sum_{l=1}^L c_l^{I} \ty(U_l)_k \;, \\
&\begin{aligned}
z^{P,i}_{k}  := \sum_{l = 1}^{L} c_l^{P,i} \ty(U_l)_k \; ,~~~ \forall  P\in \bc_R,~ i,k \in [D] \;.
\end{aligned}
\end{align}
Now from the definition of the eliminators  and \eqref{eq:z_def} we can connect the $z$ values to $\rho$ and $A.$ 
\begin{align} \label{eq: noisy measurements z}
&z^{I}_k = \sum_{\kp} \frac{A_{k \,\kp} m_{\kp \, I}}{2^{n/2}} \; , \\
&\begin{aligned}
\label{eq:zPQ_def}
z^{P,i}_{k} = z^{I}_k + s_P \sum_{ \kp} A_{k \, \kp} C_{\kp \, i} \; ,
\end{aligned}
\end{align}
where $C$ is the \emph{covariance matrix} associated with the POVM,
\vspace{-0.152in}
\begin{equation}
    C_{ij} := \BK{M_i|\bar{M}_j} = \tr(M_i M_j) - 
    \frac{\tr(M_i)\tr(M_j)}{D} \; .
\end{equation}
To emphasize, the $z-$values are obtainable from measurements. Our aim is to invert the relations in \eqref{eq: noisy measurements z} and $\eqref{eq:zPQ_def}$ to find $\rho$ and $A$ up to the unknown gauge.

Given the ability to obtain these $z$ values, the simultaneous tomography algorithm can be broken down into three steps; finding the positions of the non-zero coefficients of $\rho$ in the $\mathcal{B}_R$ basis, computing $A$ up to a gauge, and computing the other state coefficients of $\rho$  up to gauge. Below we will give a brief description of each of these steps. The full algorithm is given in \algorithmcfname~ \ref{alg:general_decoding}. Full technical details of the algorithm can be found in Appendix  \ref{app:Algo_details}.

\begin{centering}
\begin{tcolorbox}[enhanced jigsaw,
                  colback=gray!6,
                  colframe=black,
                  width=8.5cm,
                  arc=1mm, auto outer arc,
                  boxrule=0.5pt,
                 ]
{\bf  Running example:}

We see from \eqref{eq:EI_example} that $E_I = \frac{1}{4} \sum_{U \in \{I,X\}^{\otimes 2}} \Phi(U)$. So from this explicit construction, we get the $c^I_l$ values and we can compute $z^I_k$ from this.

We can group these $z-$values by the $k$ and $i$ indices into $4$ dimensional vectors and matrices.
For the $2$-qubit example we compute these values to be
$$\mathbf{z}^{I} = [0.25 ,0.25, 0.25, 0.25 ] \; ,$$
$$\mathbf{z}^{I\otimes Z}= \mathbf{z}^{Z\otimes Z} = 
\begin{pmatrix}
-0.03 &  ~0.33 &  ~0.33  & ~0.37 \\
 ~0.33 & -0.03 &  ~0.37 &  ~0.33 \\
 ~0.33 &  ~0.37 & -0.03 &  ~0.33 \\
 ~0.37 &  ~0.33 &  ~0.33 & -0.03 
\end{pmatrix} \; ,
   $$

$$ \mathbf{z}^{Z\otimes I} =
\begin{pmatrix}
 0.53& 0.17& 0.17& 0.13 \\
 0.17& 0.53& 0.13& 0.17 \\
 0.17& 0.13& 0.53& 0.17 \\
 0.13& 0.17& 0.17& 0.53 \\
\end{pmatrix} \; .
$$

We get the uniform  stochastic matrix for the other cases where $s_P = 0.$
\end{tcolorbox}
\end{centering}

\paragraph*{Step 1: Finding non-zero coefficients}\vphantom{}

\noindent To find a non-zero state coefficient, we first have to isolate all the rows of $A$  that are not all zeros.
From  \eqref{eq: noisy measurements z}, this can be clearly done by finding all $l \in [D]$ such that $z^{I}_l \neq 0.$ 
Now for one such $l$ and for $j\in [D],$ if $z_l^{P,j} - z^{I}_l \neq 0$, then $s_P$ must be  non-zero.
On the other hand if for all $j \in [D]$ if $z_l^{P,j} - z^{I}_l = 0$, then $s_P$ must be zero.

We can repeat this step for the same $l$ for each $P \in \bc_R$ to find every non-zero $s_P.$
We also store one particular $(j,l)$, obtained from  $z_l^{P,j} - z^{I}_l \neq 0$ for any $P$, to use in the final step.

\begin{centering}
\begin{tcolorbox}[enhanced jigsaw,
                  colback=gray!6,
                  colframe=black,
                  width=8.5cm,
                  arc=1mm, auto outer arc,
                  boxrule=0.5pt,
                 ]
{\bf  Running example:}
For $s_P =  0$, we will get $\mathbf{z}^P$ to be the matrix of all $0.25.$
Only $\mathbf{z}^{I\otimes Z}, \mathbf{z}^{Z\otimes I},$ and $\mathbf{z}^{Z\otimes Z}$ will differ from this and we can identify these with the non-zero coefficients of the state. 
\end{tcolorbox}
\end{centering}
\paragraph*{Step 2: Finding noise matrix up to gauge }\vphantom{}
\noindent
At this step, we choose $R \in \bc_R$ such that $s_R \neq  0$. To work around the gauge problem we have to choose one noise matrix from the one-parameter family described by \eqref{eq:gauge}.
We make this choice by taking $\alpha  = s_R$. This makes the explicit $s_R$  dependence vanish from \eqref{eq:zPQ_def}. In terms of the gauge transformed noise matrix, \eqref{eq: noisy measurements z} and \eqref{eq:zPQ_def} can be expressed as follows.
\begin{align}
&z^{I}_k  = \sum_{\kp} \frac{A'(s_R)_{k \, \kp} m_{\kp \, I}}{2^{n/2}} \; , \\
&\begin{aligned}
z^{R,i}_{k} = z^{I}_k + \sum_{ \kp} A'(s_R)_{k\,\kp} C_{\kp \, i} \; , ~~ \forall ~ i,k \in [D] \; .
\end{aligned}
\end{align}
Once we obtain the $z-$values by aggregating the measurements; we can invert the system of linear equations to find $A^\prime(s_R).$ This inversion step is always possible if the POVM is linearly independent, i.e., if the Condition \ref{cond3}
% in \eqref{eq:cond3}
holds.

\begin{centering}
\begin{tcolorbox}[enhanced jigsaw,
                  colback=gray!6,
                  colframe=black,
                  width=8.5cm,
                  arc=1mm, auto outer arc,
                  boxrule=0.5pt,
                 ]
{\bf  Running example:}
Choose the gauge to be $s_{Z\otimes I}$. Covariance matrix for computational measurements is $C_{ij} = \delta_{ij} - 0.25$. Now by plugging  $z-$values in \eqref{eq:zPQ_def} we can compute $$A'(s_{Z\otimes I}) =
\begin{pmatrix}
 0.53& 0.17& 0.17& 0.13 \\
 0.17& 0.53& 0.13& 0.17 \\
 0.17& 0.13& 0.53& 0.17 \\
 0.13& 0.17& 0.17& 0.53 \\
\end{pmatrix} \; .
$$
We can check that this matrix is indeed equal to $s_{Z\otimes I} A + (1 - s_{Z\otimes I})0.25$\;.
\end{tcolorbox}
\end{centering}

\paragraph*{Step 3: Finding state up to gauge}\vphantom{}

\noindent
In this step, we exploit the gauge transformation to find the ratio of every non-zero state coefficient with $s_R.$ From \eqref{eq: noisy measurements z} and \eqref{eq:zPQ_def} the following relation holds,
\begin{equation}
 \frac{s_P}{s_R} = \frac{ \sum_{\kp} \Ap_{k\,\kp}(s_P) C_{\kp \, i}}{ \sum_{\kp} \Ap_{k\,\kp}(s_R) C_{\kp\, i}} = \frac{ z_l^{P,j} - z^{I}_l}{ z_l^{R,j} - z^{I}_l } \; .
\end{equation}

Our choice of $(j,l)$ in \emph{Step 1} ensures that the denominator in this expression is always non-zero.

\begin{centering}
\begin{tcolorbox}[enhanced jigsaw,
                  colback=gray!6,
                  colframe=black,
                  width=8.5cm,
                  arc=1mm, auto outer arc,
                  boxrule=0.5pt,
                 ]
{\bf  Running example:}
Choose $j,l = 1$. This choice is made so that $z_j^{Z\otimes I,l} \neq z_l^{I \otimes I}.$ From the expression given above,
\vspace{-0.1in}
$$\frac{s_{I \otimes Z}}{s_{Z \otimes I}} = \frac{\mathbf{z}_{1,1}^{I\otimes Z} - \mathbf{z}_{1}^{I \otimes I}}{ \mathbf{z}_{1,1}^{Z\otimes I} - \mathbf{z}_{1}^{I \otimes I}} = \frac{-0.03 -0.25}{ 0.53 -0.25}  = -1 \; ,$$
Similarly we get,$ \frac{s_{Z \otimes Z}}{s_{Z \otimes I}} = -1.$

\end{tcolorbox}
\end{centering}

After these steps, we will know $\rho$ and $A$ up to the unknown parameter $s_R$. This unknown has to be fixed from prior information, and we will describe various ways of fixing this gauge in Section \ref{sec:fixing}.
In the next subsection, we will specialize to the case of computational basis measurements, and analyze the number of measurement shots required to implement this algorithm.

%As alluded by the definition, the set of noisy measurements obtained following the unitary operations in $\uc_{tom}$ are sufficient to perform simultaneous tomography up to the one-dimensional gauge degree of freedom. The algorithm below demonstrates this procedure.

\begin{algorithm}[!ht]
 \tcp{ \textcolor{purple}{Step 1. Find non-zero coefficients of the state}}
Compute $z_k^{I}~~ \forall ~k \in [D]$ using \eqref{eq: noisy measurements z} \\
$\mathcal{K} \leftarrow \{k \in [D]~|~z^{I}_k \neq 0\}$ \\
 $\mathcal{C} \leftarrow \{\}$ \tcp{Empty set}
 $\mathcal{S} \leftarrow [D] \times \mathcal{K}$ \tcp{Search space of index tuples}
\SetAlgoLined
\For{$P \in \bc_R$}{
    $s_P \leftarrow 0$\\
    \For{$(i,k) \in \mathcal{S}$}{
        Compute $z_k^{P,i}$ using \eqref{eq:zPQ_def} \\
        \If{$z^{P,i}_k \neq z^{I}_k$}{
            $\mathcal{S} \leftarrow \{(i,k)\}$ \\  \tcp {Replace index set with a single tuple}
            $\mathcal{C} \leftarrow \mathcal{C} \bigcup \{ P \}$ \\
            \textbf{continue}  \tcp{To the next  $P$}
        }
    }
}
 \tcp{\textcolor{purple}{Step 2. Find $A$ up to gauge symmetry}}
 \textbf{choose} $R \in \mathcal{C}$ \\
 Compute $z_k^{R,i}$ for all $k \in \mathcal{K},i \in [D]$ \\
 \For{$k \in  [D]$}{
 \If{$k  \in \mathcal{K}$}{
 Solve for $\Ap_{k\kp}$ in \\
 $\sum_{\kp}\Ap_{k\kp}C_{\kp i} = z_k^{R,i} - z_k^{I} ~~  \forall i \in [D]$ \\
 $\sum_{\kp}A^\prime_{k \kp} m_{\kp I}  =  2^{n/2} z^{I}_{k}$\\
 \Else{ $\Ap_{k\kp} \leftarrow  0 ~~ \forall \kp \in [D]$} 
 
 }}

 \tcp{\textcolor{purple}{Step 3. Find other state coefficients up to a multiplicative constant}}

$\{(j,l)\} \leftarrow \mathcal{S}$ \\
\For{$P \in \mathcal{C}$}{
$\frac{s_P}{s_R} \leftarrow \frac{z^{P,j}_l -z^{I}_l}{ z^{R,j}_l -z^{I}_l}$
}

\textbf{return} $\{ (P,\frac{s_P}{s_R}) | P \in  \mathcal{C}\},~~\Ap $
\caption{Simultaneous tomography up to gauge degree of freedom}
\label{alg:general_decoding}
\end{algorithm}

\subsection{Simultaneous tomography with randomized measurements and shot error}
\label{subsec:shoterror}
The simultaneous tomography algorithm, as described, does not consider the fact that every $\ty_k(U)$ has to be estimated using a finite number of measurement outcomes. In this section, we will specialize the algorithm to the case where measurements are made in the computational basis and analyze the number of measurement shots required to estimate the state and noise in the system up to gauge. Additionally, we will also use a randomized measurement procedure to estimate the $z$ values required for tomography. The sample complexity bounds in Theorem \ref{thm:samples} specify the number of such randomized measurements required for simultaneous tomography. This randomized measurement method can significantly reduce the overhead of simultaneous tomography as the number of operators in the tomographically complete set can be exponentially large in the system size.

We will present our results exclusively for the case of computational basis measurements, as this is the most pertinent case for practical applications.

\subsubsection{Simultaneous tomography in the computational basis}

For computational measurements the POVM is simply $\{\ket{k}\!\!\bra{k} | k \in [2^n]\}$, and the covariance operator takes a simple form, 
\begin{equation}
    C_{ik} = \delta_{ik} - \frac{1}{2^n} \; .
\end{equation}

For these types of measurements, the natural choice for the right basis set $\bc_R$ is the set of all normalized Pauli strings,
\begin{equation}
    \bc_R = \left\{\frac{I}{\sqrt2},  \frac{X}{\sqrt2}, \frac{Y}{\sqrt2}, \frac{Z}{\sqrt2} \right\} ^{\otimes n} \setminus \left\{\frac{I}{2^{n/2}}\right\} \; .
\end{equation}

%This significantly simplifies the linear system inversion required in the second step of the algorithm to find $A^\prime$

Remarkably for this choice of $\bc_R$, the effective elimination operators,  defined in \eqref{eq:EI} and \eqref{eq:EPi},  can be constructed using only Clifford operations.

 Let $\pc_X$ (or $\pc_Z$) be the set of Pauli strings composed of only $X$ (or $Z$) and $I$. Now define $H_{lQ} = \braket{l|Q|l}/2^{n/2},~\forall Q \in \pc_Z.$ Then we can show that,
\begin{align}
    E_{I} &= \frac{1}{2^n} \sum_{P \in \pc_X} \Phi(P) \label{eq:EI_pauli_def} \; ,\\
    E_{Pi} &=(1 - \sum_{Q\neq I} H_{iQ})E_{I} +  \frac{2}{2^n} \sum_{Q \neq I} H_{iQ} \sum_{\substack{Q^\prime \in \pc_X \\ [Q^\prime, Q] = 0}} \Phi(Q^\prime U_{PQ}) \; .  \label{eq:EPi_def}
\end{align}

Here $U_{PQ}$ is a member of the $n-$qubit Clifford group that maps $P$ to $Q$. Proof of this construction is given in  Appendix \ref{app:Elim_for_Pauli}.
 The overhead of applying these  eliminators  can be decreased significantly by using a randomized measurement scheme (see Appendix  \ref{app:Sample}).

 Due to the simplified nature of the covariance matrix and the POVM, the $z$ values in this setting take the following simple forms,
 \begin{align}
     z^I_k &= \frac{\sum_{k'} A_{kk'}}{2^n} \; ,\\
     z^{P,i}_k &= z^I_k + s_p(A_{ik} -z^I_k) \; .
 \end{align}

 This means that if the gauge is fixed to $s_R$, we get the following simple relation between the noise matrix and gauge $s_R$,
 \begin{equation}
 \label{eq:Comp_A_gauge}
     A'_{ik}(s_R) =  z^{Ri}_k \; .
 \end{equation}

 So for the case of computational basis measurements, the linear inversion in the second step of Algorithm \ref{alg:general_decoding} is unnecessary.

\subsubsection{Sample complexity}

The number of measurements required to estimate $\rho$ and $A$ up to a certain error depends on how far they are from violating the sufficient conditions \ref{cond1} and \ref{cond2}.
% \eqref{eq:cond1} and \eqref{eq:cond2}.
The third condition is automatically satisfied as computational basis measurements form a POVM that is linearly independent. 
More measurements are required  for simultaneous tomography the closer $\rho$ is to a maximally mixed state and the closer $A$ is to erasure channel. To measure the distance from these pathological cases, we define the following metrics for $\rho$ and $A$.
\begin{align}
\label{eq:mnorm}
    \mnorm{\rho} &\equiv \max_{P \in \bc_R \setminus I} |s_P| \; , \\
    \label{eq:unorm}
    \unorm{A} &\equiv  \max_{k,\kp \in [2^n]}   \left| A_{k\kp} - \frac{\sum_i A_{ki}}{2^n} \right| \; .
\end{align}

The sufficient conditions
% in \eqref{eq:cond1} and \eqref{eq:cond2}
\ref{cond1} and \ref{cond2}
are just non-zero lower bounds on these metrics.  These metrics allow us to state the sample complexity for the simultaneous tomography algorithm:
\begin{theorem} {\textbf{Complexity in computational basis}} \label{thm:samples} Given an $n-$qubit quantum system such that $\mnorm{\rho} > 0,~ \unorm{A} >0$. Choose a threshold parameter $0 < \beta < \mnorm{\rho}/2$. Then the three main steps of the simultaneous tomography algorithm can be implemented using randomized measurements in the computational basis with the following complexities,

\begin{enumerate}
    \item Using $O(8^n \frac{cn + log(1/\delta)}{\beta^2 \unorm{A}^2})$ randomized measurements, we can identify a non-empty subset  $\mathcal{C} \subset  \bc_R$ such that with probability $1-\delta$ the following implications hold, 
    \vspace{-0.1in}
    \begin{align*}
    P \in \mathcal{C}  &\implies |s_P| \geq  \beta \;, \\
     |s_P| \geq 1.01\beta &\implies P \in \mathcal{C} \; .
     \end{align*}
     
    \item Given $R \in \mathcal{C}$, using $O(2^n \frac{cn + log(1/\delta)}{\epsilon^2})$ randomized measurements, we can give an estimate  $\hat{A}^\prime(s_R)$ for the noise matrix up to gauge such that,
    \begin{equation*}
       Pr\left( \max_{i,j \in [D] } |\hat{A}'_{i,j}(s_R) -  A'_{i,j}(s_R) | > \epsilon \right)   \leq \delta\;.
    \end{equation*}
    
    \item  Let $\epsilon < \beta/2$. Then for a fixed $R \in \mathcal{C}$ and for every $P \in \mathcal{C}$, we can  compute an estimate  $\widehat{\dfrac{s_P}{s_R}}$  using a total of $O(2^n |\mathcal{C}|\frac{cn + \log(1/\delta)}{ \epsilon^2 \beta^2 \unorm{A}^2})$ such that,
\begin{equation*}
   Pr\left( \max_{P \in \mathcal{C}}\;\dfrac{ \left| \widehat{\frac{s_P}{s_R}} - \frac{s_P}{s_R} \right|}{ \left| \frac{s_P}{s_R} \right|}  > \epsilon  \right) \leq \delta \;.
\end{equation*}

\end{enumerate}

\end{theorem} 
    See Appendix \ref{app:Sample} for details on the randomized measurement framework used and the proof of this theorem

The three parts of this theorem correspond to the three steps of Algorithm \ref{alg:general_decoding}. In the first step of the algorithm our aim is to find the positions of the non zero coefficients of the state. The finite shot version of this step is the construction of the set $\mathcal{C}$ which is guaranteed (with high probability) to contain every operator in $\bc_R$ whose overlap with the state is greater than or equal to $1.01 \beta $ in absolute  value. Moreover it is also guaranteed with high probability that $\mathcal{C}$ will only have operators such that their overlap with the state is guaranteed to be greater than or equal to $\beta$  in absolute value. Now if we choose $\beta = \frac{1}{1.01} \min_{P:s_p \neq 0} |s_P|$, then $\mathcal{C}$ will exactly contain all the positions of the non-zero coefficients. On the other hand, if we are only concerned about recovering the noise matrix, we can take $\beta = \frac{\mnorm{\rho}}{2}$, which guarantees that $\mathcal{C}$ is non-empty. This will give us at least one non-zero coefficient to set the gauge for the problem.

The second part of the theorem concerns the estimation of the noise matrix up to gauge in the finite shot setting. This is straightforward in the computational basis as we do not have to perform any linear inversion. Notice that the sample complexity of this step is independent of $\mnorm{\rho}$  and $\unorm{A}$. This gives a considerable sample complexity advantage in the setting where we are only interested in recovering $A$. 

\begin{corollary}{\bf Estimating measurement noise (M error)}\\
Fixing $\beta = \mnorm{\rho}/2$ in Theorem \ref{thm:samples}, we can recover the noise matrix up to gauge with the error $\epsilon$, with high probability using $\tilde{O}\left(\frac{8^n}{\mnorm{\rho}^2 \unorm{A}^2} + \frac{2^n}{\epsilon^2}\right)$ randomized measurements.
\end{corollary}
Here we have used the $\tilde{O}$ notation to hide linear factors in $n$ and $\log(1/\delta)$ for the sake of readability.

The third part of the theorem concerns the estimation of state coefficients up to gauge. We only do this for $P\in \mathcal{C}$ and we set $s_P = 0$ for $P \notin \mathcal{C}$. Thus the threshold parameter ($\beta$) implicitly fixes the error in estimating these coefficients. The sample complexity of this step depends on $|\mathcal{C}|$  and hence can be considerably low if the state is sparse in the Pauli basis.

 Suppose our aim is to fully characterize the state preparation error. To estimate every element of the state with an \emph{additive} error of $\epsilon$ we show the following Corollary of Theorem \ref{thm:samples} in Appendix \ref{sec:corr_proof}.

\begin{corollary}{\bf Estimating prepared state (SP error)}\label{corr:SP}
Every coefficient of the state up to gauge can be  estimated with additive error of $\epsilon \leq \mnorm{\rho}/2$ with high probability
using a total of $\tilde{O} \left(\frac{8^n}{\epsilon^4 \unorm{A}^2} \right)$ randomized measurements
\end{corollary}

%Consolidating the sample complexities for the three steps in Theorem \ref{thm:samples}, in the worst case the total number of measurements required is seen to scale as  $O(8^n\frac{cn + \log(1/\delta)}{ \epsilon^2 \beta^2 \unorm{A}^2}).$

The exponential dependence on $n$ in these sample complexities is unavoidable because we are attempting to estimate an exponential number of independent, unknown quantities in the most general case. But this dependence can be possibly improved for special cases, like for unentangled states or binary symmetric noise channels.  We leave the analysis of such special cases for future work.

\subsection{ Numerical results}
\begin{figure*}[t]
\begin{subfigure}[t]{0.45\textwidth}
    \caption{}
    \includegraphics[width=\textwidth]{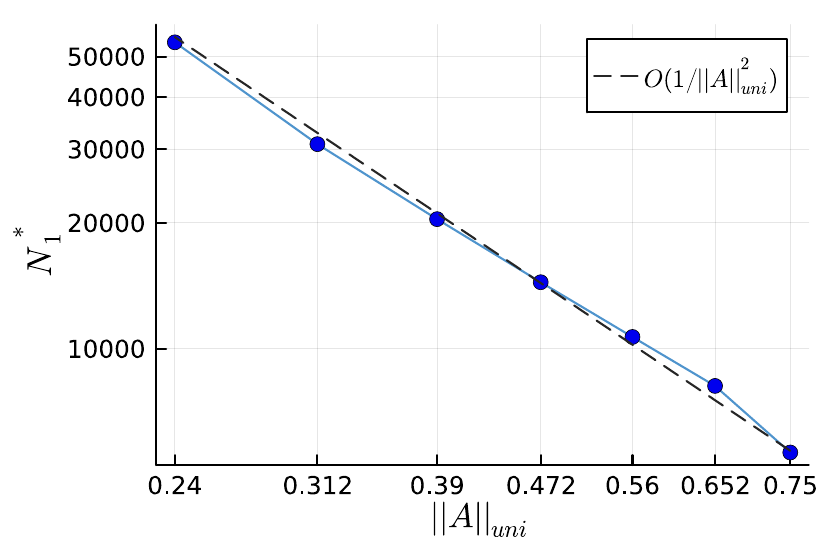}
\end{subfigure}
\begin{subfigure}[t]{0.45\textwidth}
    \caption{}
    \includegraphics[width=\textwidth]{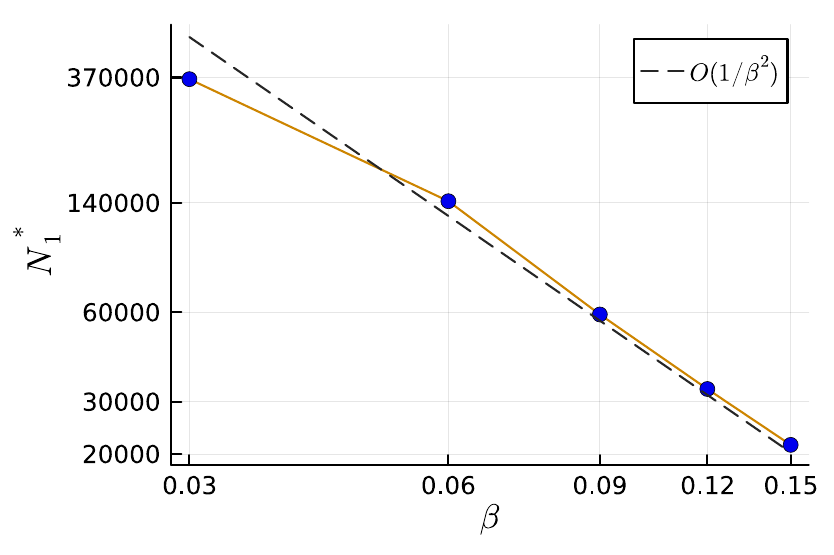}
\end{subfigure} \\
\begin{subfigure}[t]{0.45\textwidth}
    \caption{}
    \includegraphics[width=\textwidth]{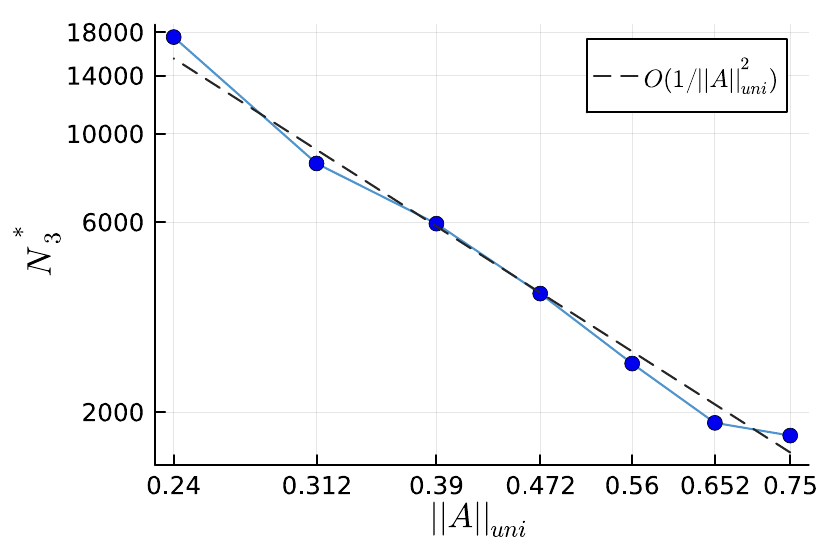}
\end{subfigure} 
\begin{subfigure}[t]{0.45\textwidth}
    \caption{}
    \includegraphics[width=\textwidth]{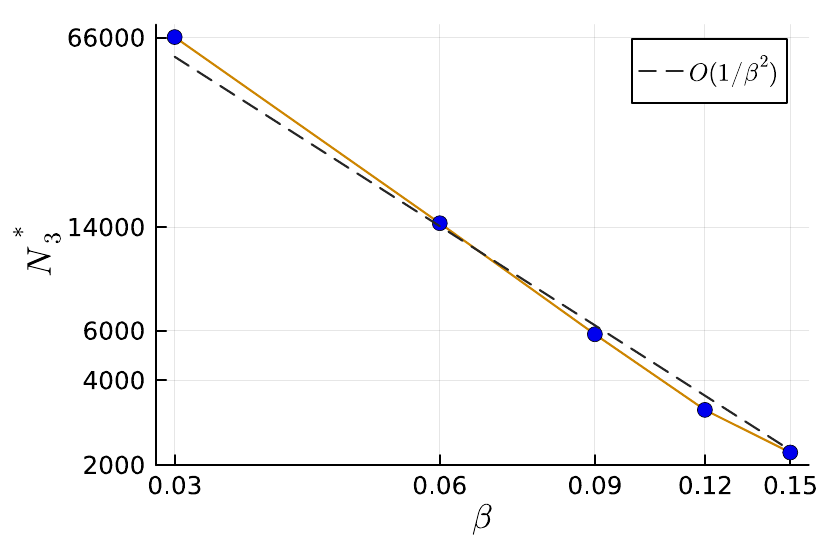}
\end{subfigure} 
\caption{{\bf Simultaneous tomography with randomized measurements in 2-qubit systems.}  {\bf (a)}, {\bf(b)} We study the scaling of Step 1 in Theorem \ref{thm:samples} for a two qubit system. $N_1^*$  here is the  number of measurements that was required to find the correct positions of the non-zero coefficients in the Pauli basis with a success rate of at least $90\%$. {\bf (c), (d)} We study the scaling of Step 3 in Theorem \ref{thm:samples}. Here $N^*_3$ is the number of measurement shots that were required to compute every state coefficient up to gauge with a multiplicative error of at most $\epsilon = 1/3$ . In { \bf (a)} and {\bf (c) }, the state is fixed to be $\rho = \ket{01}\bra{01}$ and  $A$ is chosen from the one-parameter family $A(\tau) = ( (1- \tau) I + \tau X)^{\otimes 2}$ to vary $\unorm{A}$. We fix $\beta = 1/4$. In {\bf(b)} and {\bf(d)}, the state is chosen from the one-parameter family $\rho(\tau) = \frac{I}{4} + \tau(Y\otimes I + Z \otimes Z)$ and $A$ is fixed to $A(0.05)$. $\beta$ is taken to be $\mnorm{\rho(\tau)}/2 = \tau$ and $\tau$ is varied to select $\beta.$ In all cases, $N_1^*$ and $N_3^*$ are found using logistic regression on the empirical success probability. For each value of $N$ the empirical success probability is estimated using $50$ random runs.}

\label{fig:numerical}
\end{figure*}

The tomography procedure outlined here uses independent state copies and does not use coherent measurement. In  this setting, recent results have shown a sample complexity of $\Theta(8^n)$ for noiseless tomography \cite{chen2022tight}. Hence we do not expect a substantial improvement in the $n$ dependence in Theorem \ref{thm:samples}.  More interesting is the dependence of the sample complexity w.r.t to $\unorm{A}$ and $\mnorm{\rho}$ (via the $\beta$ parameter). To check the tightness of our analysis w.r.t these quantities, we numerically study the sample complexity scaling for a few different $2-$qubit examples. The results are given in \figurename~\ref{fig:numerical}. In these experiments, we estimate the sample complexity of Steps 1 and 3 while varying $\beta$ and $\unorm{A}$ independently for a family of two-qubit states. In all the cases, we empirically observe that the number of measurement shots scale as the inverse square of these quantities, corroborating Theorem \ref{thm:samples}.

\section{Incorporating prior information: Gauge fixing and efficiency improvements}
\label{sec:fixing}
We have shown that the gauge freedom in \eqref{eq:gauge} is the only obstacle in performing simultaneous tomography. This gauge can be fixed if we have access to prior information about the state and measurement noise or access to additional measurements. Further, especially from a practical point of view, prior information can be used to significantly reduce the number of measurements and classical post-processing required to perform simultaneous tomography. The linear operator framework provides a natural way to incorporate several types of prior information. We also find that many of the example priors we use correspond to assumptions made in the error mitigation literature previously \cite{heinsoo2018rapid, nachman2021categorizing, Nation2021, vandenBerg_ModFreeReadout_2022,Lin_SPAM_2021}%, vandenBerg2023, Kim2023}. 

\subsection{Using prior information to fix the gauge}
We start with examples of prior information about $\rho$  and $A$ that are sufficient to fix the gauge. Each of these conditions imply that no two pairs $(\rho,A)$ and ($\rho',A'$) can satisfy the conditions enforced by the prior information and lie in the one-dimensional gauge manifold \eqref{eq:gauge}.\\

% and then present cases that lead to computational improvements.

\paragraph{Block independent noise:}\vphantom{}

\noindent Suppose that the POVM is described by $M_{kl} = M^1_k \otimes M^2_l$ where $k \in [D_1]$ and $l \in [D_2]$ with $D_1D_2 = D$. This can refer to a partitioning of a set of binary valued outcomes  into two parts. Suppose that the noise acts independently on the two parts such that  $A = A^1\otimes A^2$ where $A^1$ acts on $M^1$ and $A^2$ acts on $M^2$. We show that the gauge can be uniquely fixed with this information. The details of the proof of uniqueness and the algorithm to find the unique gauge are given in Appendix~\ref{app:independent_noise}. Similar uncorrelated noise models have been used as a simplifying assumption in the literature \cite{heinsoo2018rapid, nachman2021categorizing,Nation2021, vandenBerg_ModFreeReadout_2022,keith2018joint}

\paragraph{Information on purity of the state:}\vphantom{}

\noindent Let the state $\rho$ satisfy the following \emph{purity} conditions:
\begin{enumerate}
    \item     $\tr(\rho^2) = \nu $
    \item There exists $\ket{v}$ such that $\braket{v|\rho|v} > 2^{-(n-1)}$.
\end{enumerate}

If the state $\rho$ satisfies these purity conditions then the gauge can be fixed in any system of at least two qubits. In general  if the purity of $\rho$ is known, and if its min-entropy  \cite{WEHRL1991119, chehade2019quantum} is less than $n-1$, then the gauge can be fixed. 
As an important special case, we note that any \emph{pure state} satisfies the purity conditions with $\nu = 1$ and there exists an eigenvalue equal to one. 

We can use this purity information to find the gauge as follows. Algorithm~\ref{alg:general_decoding} returns a set of state coefficients $s'_P$ up to the gauge freedom such that the actual state coefficients $s_P$ are related to the ones computed by the algorithm by
 \begin{align}
     s_P = \alpha s'_P,  \quad \forall P \neq I.
 \end{align}
Using the first purity condition $\tr(\rho^2) = \nu$, we get

 \begin{equation}
    \alpha^2 = \frac{\nu - 2^{-n}}{\sum_{P, P' \neq I} {s'}_P {s'}_{P'} \BK{P|P'}} \; .
 \end{equation}
 This specifies the state up to a sign 
 \begin{align}
     \rho = \frac{I}{2^{n}} \pm \alpha \sum_{P \neq I} s'_P P \; .
 \end{align}
 Denote the two candidates by $\rho_+$ and $\rho_-$ with $\rho_+ + \rho_- = \frac{I}{2^{n-1}}$.
 Using the second purity condition, there exists a state $\ket{v}$ such that $\braket{v|\rho_+|v} > 2^{-(n-1)}$. This gives that  $\braket{v|\rho_-|v} < 0$ which in turn implies that $\rho_-$ is not a positive semi-definite matrix. Hence one of the two candidate states will not be a valid quantum state and can be used to pick the correct sign and fix the gauge.\\

\paragraph{Probe state:}\vphantom{} 

\noindent The ability to have a known state (for instance $\ket{0}^{\otimes n}$) prepared can be used to fix the gauge  as in this case $\alpha$ can be directly inferred from  \eqref{eq:gauge}.  This type of prior information is used in \cite{vandenBerg_ModFreeReadout_2022}, along with Pauli twirling for the purpose of error mitigation.

Suppose that $\rho^{pb}$ is a known probe state that is measured using the $M_k$ to give
\begin{align}
    y_k^{probe} = \sum_{\kp \in [D]} A_{k \kp} \tr(\rho^{pb} M_k).
\end{align}
Since Algorithm~\ref{alg:general_decoding} outputs a candidate $A'(\alpha)$ up to the gauge degeneracy, we have
\revision{
\begin{align}
  y_k^{probe} = \frac{1}{\alpha} \sum_{\kp} &A'_{k \kp} (\alpha) \tr(\rho^{pb} M_k) \nonumber + \\ &(1-\frac{1}{\alpha})\frac{\sum_j A'(\alpha) _{k \,j}m_{j \, I}}{2^{n/2}} \sum_{\kp} \tr( \rho^{pb} M_{\kp}).
\end{align}
}
As $\alpha$ is the only unknown quantity, it can be easily computed from the above equation.

\subsection{Using prior information for computational improvements}
In this section, we list a set of prior information that can both fix the gauge and can be used within the linear operator framework to obtain reductions in number of measurements and post-processing. \\

\paragraph{Linearly represented prior information:}\vphantom{}

\noindent We consider a set of linearly represented prior information available on the state and the noise matrix. For the state $\rho$ this refers to a set of known expectation values that may for example correspond to known physical properties of the unknown state. Using a basis we can represent this information as follows.
\begin{align}   \label{eq:linear_prior_state}
    \sum_{P \in \bc_R} s_P b_{S,P}^i = d_S^i, \quad i=1,\ldots, N_s
\end{align}
Further assume that $d_S^1 \neq 0$ which will allow us to fix the gauge.
For the noise matrix $A$, linear prior information refers to the action of $A$ on known vectors. This for example, can be used to represent knowledge about $A$ from previous experiments or from the use of multiple probe states. We denote these known quantities by
\begin{align}  \label{eq:linear_prior_noise}
    A b_{A}^i = d_A^i, \quad i=1, \ldots, N_A.
\end{align}
With access to the information in \eqref{eq:linear_prior_state} and \eqref{eq:linear_prior_noise}, we will need access to only a subspace of superoperators $\lc^{lin} \subset \lc$ to perform simultaneous tomography. Essentially we need to access information in the orthogonal subspace to those given in \eqref{eq:linear_prior_state} and \eqref{eq:linear_prior_noise} by using appropriate operators from the linear operator space. The details are given in Appendix~\ref{app:linear_information}. \\

\paragraph{Denoising and hierarchical tomography:}\label{sec: hierarchical dec} \vphantom{}

\noindent We consider a special case of linearly represented prior information on $\rho$ that provides backward compatibility with some previously designed noise-free tomography method. Suppose that a set of unitaries $\uc_{nf} = \{U_1, \ldots, U_L \}$ have been designed to perform tomography on an unknown state $\rho$ in the absence of measurement noise. The set $\uc_{nf}$ essentially encodes the prior information on $\rho$ in a linear way. This is because we assume that the state can be uniquely specified, in the noiseless setting, from the set of coefficients $\{ \tr(U\rho U^\dagger M_i)| U \in \uc_{nf}, i \in [D]  \}.$

Our goal is to utilize the set $\uc_{nf}$ and provide a method to perform simultaneous tomography in the presence of measurement noise.

In this setting, we can use $\bc_L = \bc_R = \bc = \{\Bar{M}_i \mid i \in [D]\}$ and our goal is to find the noise matrix $A$ and the coefficients of $U\rho U^\dagger$ in the basis $\bc$ for all $U \in \uc_{nf}$.
For this, we need access to a subset $\lc^{den} \subset \lc$ of linear operators given by
\begin{align}
    \lc^{den} = \{E_{ij}, \quad i,j \in [D]\} \; ,
\end{align}
where $E_{ij}$ is defined as
\begin{align} \label{eq:denoisers}
&E_{ij} \Kket{\Bar{M}_i} - \Kket{\Bar{M}_j} \in \bc_L^\perp \; ,
~~ E_{ij} \Kket{\Bar{M}_{i'}} \in \bc_L^\perp ~~~~ \forall i' \neq i \; , \nonumber \\
&E_{ij} \Kket Q \in \bc_L^\perp ~~~~ \forall Q \in \bc_L^\perp \; .
\end{align}

The set of operators in \eqref{eq:denoisers} are sufficient to \emph{denoise} each of the original measurements. If $\uc^{den}$ is a generator set for $\lc^{den}$ then the generator set for performing simultaneous tomography in the hierarchical setting is given by
\begin{align}
    \uc^{hier} = \{U_1 U_2 \mid U_1 \in \uc^{den}, U_2 \in \uc^{nf}\}.
\end{align}

Essentially, we are first using the previously designed tomography gate set $\uc^{nf}$ to prepare the states $U_1 \rho U_1^\dagger \ldots  U_L \rho U^\dagger_L,$ and then use a simultaneous tomography procedure to estimate the projection of these states into the subspace defined by the POVM.

\begin{centering}
\begin{tcolorbox}[enhanced jigsaw,
                  colback=gray!6,
                  colframe=black,
                  width=8.5cm,
                  arc=1mm, auto outer arc,
                  boxrule=0.5pt,
                 ]
{\bf  Running example:}

Suppose we have prior information that $\rho$ is diagonal in the computational basis. So if we use the computational basis measurements, $\uc_{nf} = \{I\}.$  As another example, if we know that the state, when expanded in the Pauli basis, has only $X$ terms (e.g.  $\rho \propto I + X \otimes X)$, then $\uc_{nf}$ would contain the global rotation operator that takes $X$ to $Z.$
\end{tcolorbox}
\end{centering}

\paragraph{Binary symmetric channel:}\vphantom{}

\noindent In this part, we consider a special case of hierarchical tomography where the measurements are along the computational basis. The binary symmetric channel refers to the case where each of the $n$ binary observables is flipped with a given probability independent of the rest.  
Let the bit flip probabilities be $p_1, \ldots, p_n$ where $p_i \neq 1/2$. Then the output noise matrix has the special form
\begin{align}
    A = \bigotimes_{i=1}^n A_i^{sym}, \quad \mbox{where } A_i^{sym} = \begin{bmatrix}
        1-p_i & p_i \\
        p_i & 1-p_i
    \end{bmatrix}
\end{align}
By Theorem~\ref{thm:independence_uniqueness}, the family of binary symmetric channel noise matrices allows us to fix the gauge degree of freedom. Although this is a special case of block independent noise matrtices, the extra structure can be exploited to obtain a simpler denoising algorithm. The required subset $L^{BSC} \subseteq \lc$ for performing this task is relatively small and the corresponding generator set can be realized by depth $n$ circuits consisting of $CNOT$ and $SWAP$ gates. The details of the algorithm is given in Appendix~\ref{app:BSC}. \\

\paragraph{Independent ancilla:}\vphantom{}

\noindent We consider the availability of a set of ancilla qubits where the state and measurement noise are independent of the rest. The presence of such an independent ancilla has been assumed in the work \cite{Lin_SPAM_2021}. We can view this as a special case of uncorrelated noise and therefore we can fix the gauge. The POVM is the set $\{ M^{anc}_i \otimes M_j \mid i \in [D^{a}], j \in [D^r]\}$.  The overall state and the corresponding suitable choice of basis for the traceless subspaces can be written as
\begin{align}
    \rho = \rho^{a} \otimes \rho^r, ~~ \bc_L =  \bc_L^a \otimes \bc_L^r, ~~ \bc_R =  \bc_R^a \otimes \bc_R^r,
\end{align}
where as before we choose  $\bc_L^a = \{M_i^a \mid i \in [D^a]\}$ and $\bc_L^r = \{M_i^r \mid i \in [D^r]\}$. By independence of noise on the ancilla qubits, we can decompose the noise matrix as
\begin{align}
    A = A^{a} \otimes A^r.
\end{align}
This special structure also allows us to significantly reduce the set of operators $\lc^{anc} \subseteq \lc$ that we need to perform tomography on the state $\rho$. The specification of these operators and the corresponding algorithm is given in Appendix~\ref{app:ancilla}. This serves as a partial generalization of the construction in \cite{Lin_SPAM_2021}.
~\\
We  defer a  complete analysis and optimization of all the priors discussed in this section, including sample complexities and construction of eliminators, for future work. 
\section{Discussion}\label{sec: discussion}
In this work, we have introduced a general framework for simultaneous tomography. We have completely characterized the gauge ambiguity inherent to this problem and have shown many different ways to get around this limitation. The various scenarios discussed in this context also subsume many assumptions made in prior literature to solve this problem.

There are several directions along which this work can be extended. Like any method attempting to perform full tomography on a quantum system, we find that our method also has exponential complexity in the number of qubits. Recent advances in classical shadows have given more practical methods that help in estimating accessible, but limited information from quantum states \cite{huang2020predicting}. Developing  a similar technique for simultaneous tomography would help us extract useful information from  $\rho$ and $A$. Recent works on classical shadows in the presence of noise \cite{koh2022classical} show promise in this direction. 

The ideas presented here can also be extended to the simultaneous characterization of $\rho$ along with more general forms of physical transformations acting on it. Gauge ambiguities also exist in such general cases and the same type of priors discussed here might not be able to fix the gauge. For example, the nature of the gauge transformations will change when trying to estimate  $\rho$ and a CPTP map $\Phi$ given the ability to measure states of the form of $\Phi(U\rho U^\dagger)$  \cite{Gauge2}. And in the general case, the gauge group can have a much more complicated structure than the one-parameter case discussed here. We anticipate that the present work can be possibly extended to study a much richer class of problems that naturally arise in the study of quantum systems.

\subsection*{Acknowledgements}

The authors acknowledge support from the Laboratory Directed Research and Development program of Los Alamos National Laboratory under Projects 20220545CR-NLS, 20210114ER, 20230338ER, and 20240032DR. SC was supported by the U.S. Department of Energy (DOE) through a quantum computing program sponsored by the Los Alamos National Laborator  Information Science and Technology Institute. AYL was partially supported by the U.S. DOE/SC Advanced Scientific Computing Research Program.
\bibliography{draft_bib}

\newpage

\appendix

% \onecolumngrid

% \newpage

% \begin{center}
% {\LARGE Supplementary Information}
% \end{center}
\onecolumngrid

\section{Completeness of Linear operator space}
\label{app:PTM}
In this section we show that the linear operator space induced by hybrid unitary gate operations and linear classical post-processing \eqref{eq:defining_z} is complete. 
\begin{theorem}{\bf Completeness}
\label{thm:completeness}
Let $\uc(2^n)$ be the unitary group on $n$-qubits. For any $\phi \in \mathbbm{R}^{4^n-1\times 4^n-1}$ the matrix
\begin{align}
    \begin{bmatrix}
        1 & 0 \\
        0 & \phi
    \end{bmatrix} \in \lc(\uc(2^n)).
\end{align}
\end{theorem}
To prove the theorem, we will use a basis representation where $\bc_L,\bc_R$ are the Pauli basis. However, once the completeness is proved, the result carries over to any pair of basis of the traceless space. 
The proof of the theorem relies on the existence of two constituent families of linear operators. We call the first the \emph{eliminator} operators.
\begin{lemma}[Eliminator operators] \label{lem:eliminators}
    For any $P \in \pc$ there exist operators $E_P \in \lc(\uc(2^n))$ such that
    \begin{align}   \label{eq:eliminators}
        E_P \Kket{I} = \Kket{I}, \ E_P \Kket{P} = \Kket{P} \nonumber, \\
        E_P \Kket{Q} = 0 \ \text{for all} \ Q  \in \pc \setminus \{P,I \}.
    \end{align}
\end{lemma}
The second corresponds to unitary transformations that take one Pauli operator to another. We call these \emph{permutation } operators.
\begin{lemma}[Permutation operators] \label{lem:permutation}
    For any $P,Q \in \pc$ with $P \neq Q$, there exists a unitary $U_{PQ}$ such that $\Phi(U_{PQ})\Kket{P}  = \Kket{Q}$. Moreover $U_{PQ}$ can be implemented using a circuit of at most $O(n)$ depth.
\end{lemma}

\begin{proof}[Proof of Theorem~\ref{thm:completeness}]
    We prove the theorem by constructing a set of canonical linear operators. For $P,Q \in \pc$ let $e_{PQ} \in \mathbbm{R}^{4^{n}-1\times 4^{n}-1}$ such that $[e_{PQ}]_{P,Q} = 1$ and $[e_{PQ}]_{P',Q'} = 0$ whenever $(P',Q') \neq (P,Q)$. Let $\alpha \in \mathbbm{R}$ and define the linear operator $E_{PQ}^{\alpha}$ given by 
    \begin{align}
        E_{PQ}^{\alpha} = \begin{bmatrix}
        1 & 0 \\
        0 & \alpha e_{PQ} \end{bmatrix}
    \end{align}
    We will show that $E_{PQ}^{\alpha} \in \lc$ using an explicit construction. Showing this is sufficient to prove the theorem as every super-operator that preserves $\Kket{I}$ can be written as a linear combination of these eliminators, such that the coefficients sum to unity.

    Using Lemma~\ref{lem:permutation} and since $U_{PQ}$ is unitary, we get $\phi(U_{PQ}) \Kket{P} = \Kket{Q}$ and for any
    $P' \neq P$ the orthogonality condition $\tr(\phi(U_{PQ}) \Kket{P'}, Q) = 0$. Therefore, using the eliminator operators in Lemma~\ref{lem:eliminators} we get $E_{Q} \phi(U_{PQ}) \Kket{P'}= 0$ for all $P' \neq P$. So the operator $E_{PQ}^{1}$ can be explicitly constructed as
    \begin{align}
        E_{PQ}^{1} =  E_Q \phi(U_{PQ}).
    \end{align}
    By closedness under composition from Lemma~\ref{lem:closed}, we get $E_{PQ}^{1} \in \lc$.
    For arbitrary $\alpha \in \mathbbm{R}$, we can construct
    \begin{align}
        E_{PQ}^{\alpha} := \alpha (E_{PQ} - E_{I}) + E_{I}.
    \end{align}
    It follows that $E_{PQ}^{\alpha} \in \lc$ from closedness under linear combination in Lemma~\ref{lem:closed_linear}. Finally, any operators can be constructed as 
    \begin{align}
        \begin{bmatrix}
            1 & 0 \\
            0 & \phi
        \end{bmatrix}  = \sum_{P,Q \in \pc \setminus I} \frac{1}{(4^n-1)^2} E_{PQ}^{(4^n -1 )^2\phi_{PQ}} .
    \end{align}
    
\end{proof}

We now establish the existence of the \emph{eliminators} and \emph{permutation} operators.
For $P \neq I$, let $F_P = E_P- E_{I}$. These are rank-1 projectors to the operator $P$. For consistency, we will also take $F_I =E_{I}$.
\begin{proof}[Proof of Lemma~\ref{lem:eliminators}]
    For the case of one qubit, we can easily verify that the operators
    \begin{align}
        F_I^1 &= E_{I}^1 =  \frac{1}{4}(\phi(I) +\phi(X) +\phi(Y)+\phi(Z)) \; , \\
        F_X^1 &= \frac{1}{4}(\phi(I) + \phi(X) - \phi(Y) - \phi(Z)) \; , \\
        F_Y^1 &= \frac{1}{4}(\phi(I) + \phi(Y) - \phi(X) - \phi(Z)) \; , \\
        F_Z^1 &= \frac{1}{4}(\phi(I) + \phi(Z) - \phi(X) - \phi(Y)) \; ,
    \end{align}
    satisfy \eqref{eq:eliminators}. For the n-qubit case when $P = \otimes_{i=1}^n P_i$, we can construct the n-qubit projector operator as $F_{P} = \otimes_{i=1}^n F_{P_i}^1$.
    Now $E_P = F_I + F_P = \frac{1}{4^n}\sum_{P'} \phi(P') + F_P$. It is clear that the expansion of $F_P$ in terms of superoperators $\Phi(P')$ will have positive coefficients if and only if $P'$ commutes with $P$. This is because for $P'$ and $P$ to commute they should have either zero or an even number of anti-commuting  pairs of single qubit operators when they are matched qubit-wise.
    From this argument, it is clear that $E_P = \frac{2}{4^n}\sum_{P': [P',P] = 0} \phi(P')$
\end{proof}
This is akin to what is known as \emph{twirling} in the literature \cite{dankert2009exact}.

\begin{proof}[Proof of Lemma~\ref{lem:permutation}]
    Circuits that map between Pauli strings from the Clifford group are well studied in the literature. These circuits are examples of stabilizer circuits, i.e.  they can be constructed using only CNOT, Hadamard, and Phase gates \cite{gottesman1998theory}. Any such stabilizer circuit can be constructed using only $O(n)$ depth \cite{maslov2018shorter}.
\end{proof}

\section{Identifiability: Proof of Theorem~\ref{thm:degeneracy_char}}
\label{app:gauge}
Let $s_P, s_P', $ with $ P \in \pc^n$, denote the coefficients of $\rho$ and $\rho'$ in the Pauli basis. Then using the assertion of the theorem and \eqref{eq:noisy_observations}, we have for all $k \in [D]$,
\begin{align}   \label{eq:zero_coeffs}
    2^{-n/2}&\sum_{\kp \in [D]} (A_{k\kp}-A'_{k\kp})m_{\kp I} \ + \nonumber \\
    &\sum_{\substack{\kp \in [D], \\ P,Q \in \pc\setminus I} } (s_P A_{k \kp} - s'_P A'_{k \kp}) \phi(U)_{PQ}    m_{\kp Q} = 0,
\end{align}
for all $U \in \uc$. Thus for any set of unitary operators $U_1, \ldots, U_L$ and scalars $c_1, \ldots, c_L$ such that $\sum_{l}c_l = 1$,
\begin{align}   \label{eq:zero_coeffs_1}
    2^{-n/2}&\sum_{\kp \in [D]} (A_{k\kp}-A'_{k\kp})m_{\kp I} \ + \nonumber \\
    &\sum_{\substack{\kp \in [D], \\ P,Q \in \pc\setminus I} } (s_P A_{k \kp} - s'_P A'_{k \kp}) \phi_{PQ}    m_{\kp Q} = 0,
\end{align}
where $\phi  = \sum_l c_l \phi(U_l) \in \lc(\uc)$. Since \eqref{eq:zero_coeffs_1} holds for any $\phi \in \lc(\uc)$, and by Theorem~\ref{thm:completeness}, the linear operator $ \begin{bmatrix}
        1 & 0 \\
        0 & 0
    \end{bmatrix} \in \lc(\uc)$, 
we must have  
\begin{align} \label{eq:zero_I_coeff}
\sum_{\kp} A_{k \kp} m_{\kp I} = \sum_{\kp} A'_{k \kp} m_{\kp I}.
\end{align}
    Similarly, using Theorem~\ref{thm:completeness} we get that $ \begin{bmatrix}
        1 & 0 \\
        0 & e_{PQ}
    \end{bmatrix} \in \lc^n$ for all $P \in \bc_R$ and $Q \in \bc_L$. An identical argument yields for all $ k \in [D]$,
\begin{align}
    \sum_{\kp} (s_P A_{k \kp} - s'_P A'_{k \kp}) m_{\kp Q} = 0,
\end{align}
or equivalently in matrix form,
\begin{align}
    (s_P A - s'_P A') \Bm_{\setminus I} = 0, \quad \forall P \in \bc_R.
\end{align}
\revision{Using this with Lemma~\ref{lem:left_null_space}, we can see that the  matrix      $(s_P A - s'_P A')$ must have the following form,
\begin{align}
    s_P A_{kk'} - s'_P A'_{kk'} =  -d_k.
\end{align}
}
\revision{Here $d_k$ is as of yet undetermined. In the following steps, we will fix the value of $d_k$ from \eqref{eq:zero_I_coeff}.
Defining $s_P/s'_P = \alpha$, we get 
\begin{align}
    A' = \alpha A + \frac{\text{diag}(d)}{s'_P} \mathbbm{1}.
\end{align}
Combining with \eqref{eq:zero_I_coeff} and using the fact that $\sum_{k}m_{kI} = 2^{n/2},$ we get,
\begin{align}
    \sum_{\kp} A'_{k \kp} m_{\kp I} = \alpha \sum_{\kp} A_{k \kp} m_{\kp I} + \frac{d_k}{s'_P} 2^{n/2}, 
\end{align}
which gives
\begin{align}
    d_k = \frac{s'_P(1-\alpha)\sum_{\kp} A_{k \kp} m_{\kp I}}{2^{n/2}}.
\end{align}}
This completes the proof of Theorem~\ref{thm:degeneracy_char}.

\section{Making a POVM linearly independent}
\label{app:new_POVM}
\begin{lemma}

Given a $D$-outcome POVM such that the linear span of this POVM has only dimension $r$, then we can always construct an $r-$outcome , linearly independent POVM by taking linear combinations of the original POVM elements. \end{lemma}
\begin{proof}

Our aim is to construct a new $r-$ outcome POVM,
\begin{equation} \label{eq:M_new_1}
M'_j = \sum_{i=1}^D p_{ji} M_i,~~j \in [r].
\end{equation}
If the matrix $p$ is full rank (rank $r$). Then it is easy to check that $M'_j$ is satisfies condition \ref{cond3}. Moreover the outcomes probabilities corresponding to the new POVM can calculated from the outcome probabilities of the old POVM.

Because of the linear dependence between the POVM elements, there exists $D-r$  independent vectors ($\vec{c}_i$) \footnote{$\vec{u}$ notation hides the superscript index which runs from $1$ to $D$.}
 in $\mathbb{R}^D,$ such that $\sum_k c_i^k M_k = 0.$ Now take some PSD matrix $O$ with nonzero overlap with all the POVM operators. Define $u^k = \tr(M_k O)$: it follows $\langle\vec{u}, \vec{c}_i \rangle = 0$. Thus there exists one vector ($\vec{u}$) with positive coefficients that is orthogonal to every $\vec{c}_i$.
 
Now let $F$ be the $r$ dimensional subspace of $\mathbb{R}^D$ consisting of all the vectors orthogonal to the set $\{\vec{c}_i| i = 1,\ldots, D-r \}.$ We know that $\vec{u} \in F$. We now claim that $F$ can be spanned by a set of $r$  linearly independent, \emph{positive vectors}.

Suppose $\vec{u}, \vec{v_1} \ldots\vec{v}_{r-1}$ is a linearly independent set that spans $F$.  We can always chose a positive $\alpha_i$ such that $\vec{v}'_i :=  \vec{v}_i + \alpha_i \vec{u}$ is positive. Now we need to prove that these newly defined positive vectors are linearly independent.

Suppose there was some linear dependence between the new vectors, such that $\vec{u} + \sum_i\vec{v}_i' a_i = 0 $. This would imply a linear relation between the original $\vec{v}_i$ and $\vec{u}$. This would obviously contradict the initial statement that this is a linearly independent set that spans $F.$  

Thus we can always construct a set of linearly independent, positive vectors $\{ \vec{u}, \vec{v}_1', \ldots \vec{v}'_{r-1} \}$  that span $F$. 
For purely notational convenience define,
\vspace{-0.07in}
\begin{align*}
\vec{\mu}_1 &:= \vec{u}  \; , \\
 \vec{\mu}_{j} &:= \vec{v}'_{j-1} \; ,~~j = 2 \ldots r \; .
 \end{align*}
With these vectors as coefficients we can recombine the old POVM operators to construct an $r$-outcome POVM.,
\begin{equation}\label{eq:M_new_2}
 M'_j = \sum_{i =1}^D \mu^i_j  \frac{M_i}{\sum_l \mu^i_l} \; ,~~j \in [r] \; . 
\end{equation}
The positivity of $\vec{\mu}_i$ ensures that these new operators are positive. The normalization used in this relation ensures that the new operators sum to identity.
Now comparing \eqref{eq:M_new_1} and \eqref{eq:M_new_2}, we see that the desired $p$ matrix is, $p_{ji} = \frac{\mu_j^i}{\sum_l \mu^i_l}.$ By construction the matrix of $\mu_j^i$ has rank $r$. Matrix $p$ is obtained by normalizing all the columns of the  $\mu$ matrix. Since this cannot change the number of independent columns, $p$ must also have rank $r$. Which in turn implies that $\{M'_j | j = 1,\ldots,r \}$  satisfies the condition \ref{cond3}.
\end{proof}

\section{Detailed description of Algorithm \ref{alg:general_decoding}}
\label{app:Algo_details}
In Step 1, we want to find a non-zero state coefficient to set this as the gauge for the entire problem. From \eqref{eq:gauge}, it is clear that $s_P$ must be non-zero to act as a valid gauge parameter. To this end we must first find all rows of $A$ that are not all zeros. This is done by checking that $z_k^{I} \neq 0$, as this will imply that the $k-$th row of $A$ is non-zero. Using this we can construct a set $\mathcal{K}$ that holds the location of all non-zero rows. Since,
$$
z^{P,i}_{k} - z^{I}_k = s_P \sum_{ \kp} A_{k\kp} C_{\kp i} \; , ~~ \forall  P\in \bc_R \; ,~ i,k \in [D] \; ,$$
if $ z^{P,i}_{k} - z^{I}_k = 0$, that can either be due to $s_p =0$ or $\sum_{\kp} A_{k\kp} C_{\kp i} =  0 $.

Now to eliminate the second case, we check if $  z^{P,i}_{k} - z^{I}_k = 0$ for all values of $i$ and $k \in \mathcal{K}$. If this is the case and $s_P \neq 0$; then that is only possible if all the non-zero rows of $A$  are in the null space of $C.$ But if the POVM is independent then $C$ has a rank of $D-1$ (see Lemma \ref{lem:rank_of_cov}). It is easy to check from the definition of $C$ that the null space consists only of the uniform vector. So if all the non-zero rows of $A$  are uniform vectors, then $A$ is an erasure channel and it violates Condition \ref{cond1},
% (Eq.~(\ref{eq:cond1}))
necessary for simultaneous tomography. So assuming that this condition holds, $s_P \neq 0$ must imply that there exist some $j,l$ such that  $ z^{P,j}_{l} - z^{I}_l \neq 0.$. Note that this pair $j,l$ is such that $\sum_{\kp}A_{l\kp} C_{\kp j} \neq 0.$ Once found for some $P$ can be reused to check other state-coefficients as it is state-independent. We also store this $j,l$ for the last step of the algorithm.

In Step 2, we pick an $s_R$  from $R \in \mathcal{C}$ as the unknown gauge for our problem. In terms of $A'(s_R)$, we have the equations,
$$ z^{R,i}_{k} - z^{I}_k = \sum_{ \kp} A'(s_R)_{k\kp} C_{\kp i},~~ \forall  P\in \bc_R,~ i,k \in [D] \; , $$
and $$z_k^{I} =  \sum_{\kp} \frac{A'(s_R)_{k \kp} m_{\kp I}}{2^{n/2}} \; . $$

Now since the columns of $C$ along with the uniform vector form a full rank system in $\mathbb{R}^D$, we can invert this system of equations to find $A'(s_R)$.

In Step  3, we use the gauge transform equations \eqref{eq:gauge} again to find relations between state coefficients
\revision{
$$ A_{k \kp}'(s_R) = \frac{s_R}{s_P}A_{k \kp}'(s_P) + ( 1 - \frac{s_R}{s_P})\frac{\sum_{\kp} A'_{k i}(s_P)m_{i I}}{2^{n/2}} \; .$$}

\begin{lemma}
\label{lem:rank_of_cov}
The $D \times D$ covariance matrix has rank $D-1$ if the POVM is linearly independent. And the null space of the covariance matrix is spanned by the uniform vector.
\end{lemma}

\begin{proof}
The Covariance matrix can be equivalently expressed as,
\begin{equation}
    C_{ij} = \langle \braket{\bar{M_i}|\bar{M_j}}\rangle \; .
\end{equation}

This is the Gram matrix associated with the set of traceless measurement operators, $C = W^\dagger W$, where $W = \sum_i \Kket{\bar{M}_i}\!\!\bra{i}.$ From the definition of Gram matrix, it is clear that the rank of $C$ will be equal to the dimension of the subspace spanned by the traceless measurement operators.
Now the space spanned by the measurement  operators has dimension $D$, by the independence assumption. The normalization constraint on the measurement operators gives,
\begin{equation}
    \sum_i M_i = I ~~ \implies \sum_i \bar{M}_i  = 0 \; .
\end{equation}

This implies that the subspace spanned by the traceless operators has dimension $D-1$ or less. Suppose that there was some other set of coefficients $c_i$ such that $\sum_i c_i \bar{M}_i = 0$. This would then imply,

\begin{align}
\sum_i c_i M_i - \left(\sum_j c_j \langle\braket{M_j|I}\rangle \right) \frac{I}{D} = 0\\
\implies \sum_i \left(c_i -  \frac{\sum_j c_j \langle\braket{M_j|I}\rangle}{D} \right)M_i = 0 \; .
\end{align}

This is only possible if all the $c_i$ are the same. So the only possible linear relation between the traceless measurements is $\sum_i  \bar{M}_i = 0.$ This fixes the dimension of their span, and the rank of the covariance matrix to be $D-1.$

The null space of $C$ is spanned by the uniform vector as $\sum_j C_{ij} = 0.$

\end{proof}

\section{Elimination operators for the Pauli basis and computational measurements}
\label{app:Elim_for_Pauli}
For computational basis measurements, we take $\bc_R$ to be  the set of all normalized  Pauli operators on $n-$qubits, $\hat{\pc}^n = \{I/\sqrt2,X/\sqrt2,Y/\sqrt2,Z/\sqrt2 \}^{\otimes n}$

\begin{equation}
\label{eq:br_pauli}
    \bc_R = \hat{\pc}^n \setminus \{I/\sqrt{2^n}\} \; .
\end{equation}

For $\bc_L$ we take the traceless part of the computational basis POVM
\begin{equation}
\label{eq:bl_pauli}
   \bc_L = \{\bar{M}_k =  \ket{k}\!\!\bra{k} - \frac{I}{2^n} ~| ~k  \in [2^n] \} \; .
\end{equation}

From this definition, it is clear that  $\bc_L^\perp$ will be the space of all fully off-diagonal operators in the computational basis.

We also denote by $\pc_X$ $(\pc_Z)$, the set of all Pauli strings consisting of only $X (Z)$ operators and $I$.
As before the normalized matrices are given by  $\hat{Q} = Q/\sqrt{2^n}.$

Now using these definitions we prove the relations in \eqref{eq:EI_pauli_def} and \eqref{eq:EPi_def}.
\begin{lemma}
Let $\bc_R$ and $\bc_L$ be as defined in \eqref{eq:br_pauli} and \eqref{eq:bl_pauli}. Let $E_{I}$ be the $n-$qubit superoperator defined  as in \eqref{eq:EI}. Then,
$$E_{I} = \frac{1}{2^n} \sum_{P \in \pc_X} \Phi(P) \; ,$$ where $\pc_X = \{I,X\}^{\otimes n} \; .$
\end{lemma}

\begin{proof}
First let us look at the one-qubit version of this superoperator, $E^{1}_{I} = \frac{\Phi(I) + \Phi(X)}{2}.$ It is clear that this superoperator satisfies all the conditions outlined in \eqref{eq:EI}. It leaves the identity invariant and maps the Pauli operators to either zero or an off-diagonal operator.
\begin{align*}
E^{1}_{I}\Kket{I} = \Kket{I} \; , ~~ E^{1}_{I} \Kket{X} &= \Kket{X} \; , \\
E^{1}_{I}\Kket{Z} = 0 \; , ~~ E^{1}_{I} \Kket{Y} &= 0 \; .
\end{align*}
Now by elementary linear algebra, the following relation holds 
\begin{equation}
    \Phi(A) \otimes \Phi(B) = \Phi(A \otimes B) \; .
\end{equation}
This is because for any operators $V$ and $W$ 
\begin{align*}
\Phi(A) \otimes \Phi(B) \Kket{V} \otimes \Kket{W} &= (AVA^\dagger) \otimes (BWB^\dagger) \\
&= (A \otimes B) (V \otimes W) (A \otimes B) ^\dagger \\ 
&=  \Phi(A \otimes B) \Kket{V} \otimes \Kket{W} \; . 
\end{align*}
From this elementary fact,
$ E_{I} = \frac{1}{2^n} \sum_{P \in \pc_X} \Phi(P)= \left( \frac{\Phi(I) + \Phi(X)}{2} \right)^{\otimes n} = (E^1_{I})^{\otimes n}$.
This implies that $E_{I} $ acting on any operator in $\bc_R$ will either give zero or an operator in $\pc_X.$  On the other hand, it is clear that $E_{I}$ leaves the identity invariant. Thus $E_{I} = \frac{1}{2^n} \sum_{P \in \pc_X} \Phi(P) $ is consistent with the definition in \eqref{eq:EI}.     
\end{proof}

Now, to prove the same for $E_{Pi}$ eliminators defined in \eqref{eq:EPi}, first we define some auxiliary eliminators. These eliminators effectively map between normalized Pauli strings.
Let $P \in \bc_R$  and $Q \in \pc_Z  \setminus \{I\}$. 
\begin{align}
\label{eq:EPQ}
&E_{PQ} \Kket{I} = \Kket{I} ,~E_{PQ} \Kket{P} -  \Kket{\hat{Q}} \in \bc_L^\perp \; , \nonumber \\
&E_{PQ} \Kket{P^\prime} \in \bc_L^{\perp}~~~~ \forall  P^\prime \in \bc_R \setminus \{ P\} \; .
\end{align}
Using these we can write $E_{Pi}$ as follows,
\begin{equation}
\label{eq:Pi_from_PQ}
E_{Pi} = (1 - \sum_{Q \in \pc_Z \setminus \{I\}} H_{iQ})E_{I} + \sum_{Q \in \pc_Z \setminus \{I\}} H_{iQ} E_{PQ} \; ,
\end{equation}
where $H_{iQ} = \frac{\braket{i|Q|i}}{\sqrt{2^n}} = \tr(\bar{M}_i \hat{Q}) $.
Since $\pc_Z \setminus {I}$ is an orthogonal basis for traceless diagonal matrices, we have $\bar{M}_i = \sum_{Q \in \pc_Z \setminus \{I\} } H_{iQ} \hat{Q}.$
From this, it is easy to check that \eqref{eq:Pi_from_PQ} holds.

\noindent Now, if $U_{PQ}$ is a member of the Clifford group on $n$-qubits that maps from $P$ to $Q$, then $E_{PQ} = E_{QQ}\Phi(U_{PQ}).$ 
So to show the relation in \eqref{eq:EPi_def}, we only need to show that 
$E_{QQ} = \frac{2}{2^n} \sum_{\underset{Q': [Q,Q'] = 0}{Q' \in \pc_X}} \Phi(Q').$
\begin{lemma}
Let $\bc_R$ and $\bc_L$ be as defined in \eqref{eq:br_pauli} and \eqref{eq:bl_pauli}. Let $E_{QQ}$ be the $n-$qubit superoperator defined  as in \eqref{eq:EPQ} for all $Q \in \pc_Z \setminus {I}$. Then,
$$E_{QQ} = \frac{2}{2^n} \sum_{\underset{Q': [Q,Q'] = 0}{Q' \in \pc_X}} \Phi(Q') \; ,$$ where $\pc_X = \{I,X\}^{\otimes n} \; .$
\end{lemma}
\begin{proof}
The proof idea here is similar to that used in Lemma \ref{lem:eliminators}.

For $n =1$ we can easily verify that $E_{ZZ} = \Phi(I)$.
Let $F_{QQ} = E_{QQ}-E_I$. For the single qubit case  we get $F_{ZZ} = \frac{1}{2}(\Phi(I) - \Phi(X))$.
For consistency we will denote $F_{II} = E_{I}.$

\noindent Now, for any $n-$qubit $Z$ string $Q = Q_1 \otimes \ldots Q_n$, we can easily check that $F_{QQ}= \otimes_{i =1 }^n F_{Q_i Q_i}.$ Since $E_{QQ} = E_I + F_{QQ} $, the operator strings that remain in the expansion for $E_{QQ}$ will be those which have an even number of single qubit $X$ operators at positions where $Z$ operators are present in $Q$. In the positions where $Z$ does not exist in $Q$, the operators in the expansion can either be $X$ or $I$. This precisely describes all the operator strings in $\pc_X$ that commute with $Q.$
\end{proof}

\ca

For larger $n$ we can give proof by induction. Before  getting into the inductive argument let us define a few quantities. First are subsets of $\pc_X$ that commute and anti-commute with $Q$
\begin{align}
\scc^n_X(Q) = \{Q'| Q' \in \pc_X, [Q',Q] = 0 \} \\
\tsc^n_X(Q) = \{Q'| Q' \in \pc_X, \{Q',Q\} =0 \}
\end{align}

Obviously $\scc^n_X(Q) \bigcup \tsc^n_X(Q) = \pc_X.$

Now define two superoperators corresponding to these sets,

\begin{align}
S_Q &= \frac{2}{2^n} \sum_{Q' \in \scc^n_X(Q)} \Phi(Q')\\
\tilde{S}_Q &= \frac{2}{2^n} \sum_{Q' \in \tsc^n_X(Q)} \Phi(Q')
\end{align}

Our aim is to prove that $S_Q = E_{QQ}$. We can see that $\tilde{S}_Q = 2E_{I} - S_Q.$

Now let us state the inductive argument:
\emph{Assume $E_{PP} = S_P$ for any $n-$qubit Pauli $Z$ string. Then the same is true for all $n+1$ -qubit Pauli $Z$ strings.}

To show this we will consider two cases \\
{\bf Case 1} $Q = P \otimes Z$ \\
In this case the right-most operator in $Q$ a single qubit $Z$. Here $P$ is an $n$-qubit Pauli $Z$ string.

For this $Q$ it is easy to see that,
\begin{equation}
   \scc_X^{n+1} (Q) = (\scc^n_X(P) \otimes {I}) \bigcup  (\tsc^n_X(P) \otimes {X})
\end{equation}

 Here we define, $\scc^n_X(P) \otimes {I} \equiv \{P'\otimes I| P \in \scc^n_X(P) \} .$

From this relation, using the assumption in the inductive argument we have,
\begin{align}
S_Q &= \frac{2}{2^{n+1}} \left( \sum_{P' \in  \scc^n_X(P)} \Phi(P' \otimes I ) +  \sum_{B' \in  \tsc^n_X(P)} \Phi(B' \otimes X ) \right),  \\
%&= \frac{2}{2^{n+1}} \left( \sum_{P' \in  \scc^n_X(P)} \Phi(P') \otimes \Phi(I ) + \sum_{B' \in  \tsc^n_X(P)} \Phi(B') \otimes\Phi(X) \right),  \\ 
&= \frac{1}{2} \left( S_P \otimes \Phi(I) + \tilde{S}_P \otimes \Phi(X) \right)\\
&= \frac{1}{2} \left( E_{PP} \otimes \Phi(I) + (2E_{I} - E_{PP}) \otimes \Phi(X) \right)
\end{align}

From this relation, we can check that $S_Q \Kket{I} = \Kket{I}$.
Now let $E_{PP}\Kket{P} = P + B$, where $B$ is a fully off-diagonal operator.
This gives us,
$$S_Q\Kket{Q} = S_Q\Kket{P} \otimes \Kket{Z} = Q + B \otimes Z$$
This tells us that $S_Q\Kket{Q} - Q \in \bc^\perp_L$

Now for $P' \neq P$, we can verify $S_Q \Kket{P'} \otimes \Kket{Z}  \in \bc^\perp_L $ and $S_Q \Kket{P'} \otimes \Kket{I}  \in \bc^\perp_L $.
Also we can check that $S_Q \Kket{P}\otimes\Kket{I} \in \bc^\perp_L.$
These checks clearly exhausts all possible operator $n+1$ qubit Pauli $Z$ strings and show that $S_Q = E_{QQ}.$\\

{\bf Case 2:  $Q = P \otimes I$ }\\
This is simpler compared to the previous case
In this case,

\begin{equation}
   \scc_X^{n+1} (Q) = \scc^n_X(P) \otimes \{I,X\}
\end{equation}

This gives,
\begin{equation}
   S_Q = S_P \otimes E_{I}  = E_{PP} \otimes E_{I}
\end{equation}

We can check explicitly by checking various inputs as before that $E_{PP} \otimes E_{I} = E_{QQ}.$

\cb

\section{Sample complexity of  simultaneous tomography}
\label{app:Sample}

\subsection{Randomized version of simultaneous tomography in the computational basis.}

First we describe how to use a finite number of measurement shots to estimate the $z-$values.

\paragraph{Finite sample estimator for $z_k^I$} \vphantom{}
\label{para:est_1}

\noindent Choose an $n$ qubit Pauli string uniformly at random from $\pc_X.$ Now apply this operator to the quantum state and measure the outcome. Repeat this procedure $N$ times using independent copies of the state and record the $N$ measurement outcomes.

\noindent Let now $X^l_1, \ldots X^l_N$ be binary random variables such that $X^l_i$ records whether the $i-$the measurement outcome is $l$. We then define the estimate
\begin{equation}
\label{eq:est_1}
\hat{z}^I_l := \frac{1}{N} \sum_i X^l_i \;.
\end{equation}
This is an unbiased estimate for $z^I_l$
 \begin{align*}
 \mathbb{E} X^l_i &=  Pr(X^l_i = 1) = \frac{1}{2^n} \sum_{P \in  \pc_X} Pr(X^l_i = 1 | P) \\
 &=  \frac{1}{2^n} \sum_{P \in  \pc_X}  \ty_l(P) = z^I_l \; . 
 \end{align*}
Notice that the same $N$ measurement outcomes can be processed in different ways to get estimates $z_l^I$ for all $l \in [2^n]$.

\paragraph{ Finite sample estimator for $z_k^{P,i}$}\vphantom{}
\label{para:est_2}

\noindent Let us define
\begin{equation}
    z_l^{PQ} :=   \frac{2}{2^n} \sum_{\substack{Q^\prime \in  \mathcal{P}_X \\ [Q^\prime,Q] = 0}}  \ty_l(Q^\prime U_{PQ}) \; ,
\end{equation}
where $U_{PQ}$ is a member of the Clifford group that maps $P$ to $Q$. From \eqref{eq:EPi_def}
$$ z_l^{P,i} = (1- \sum_{Q \in \pc_Z \setminus {I} } H_{iQ}) z_l^{I} +  \sum_{Q \in \pc_Z \setminus {I} } H_{iQ} z^{PQ}_l \; .$$
We can estimate $z_l^{PQ}$ in the following way.
Choose an $n$ qubit Pauli string $Q^\prime$ uniformly at random from with $Q$. Since half of $\mathcal{P}$ commutes with $Q$ this can be done with constant overhead. Now apply $Q^\prime U_{PQ}$ to the quantum state and measure the outcome. Repeat this procedure $N^\prime$ times using independent copies of the state and record the $N^\prime$ measurement outcomes.

\noindent Let $X^l_1, \ldots X^l_{N^\prime}$ be binary random variables such that $X^l_i$ records whether the $i-$the measurement outcome is $l$. Then we can define the following estimates
\begin{equation}
\hat{z}^{PQ}_l := \frac{1}{N^\prime} \sum_i X^l_i \;,  \label{eq:est_2}
\end{equation}
\begin{equation}
\hat{z}_l^{P,i} := (1- \sum_{Q \in \pc_Z \setminus {I} } H_{iQ}) \hat{z}_l^{I} +  \sum_{Q \in \pc_Z \setminus {I} } H_{iQ} \hat{z}^{PQ}_l \; . \label{eq:est_3}
\end{equation}
These are also unbiased estimates since
  \begin{align}
 \mathbb{E} X^l_i &=  Pr(X^l_i = 1) = \frac{2}{2^n} \sum_{\substack{Q^\prime \in  \mathcal{P}_X \\ [Q^\prime,Q] = 0}}   Pr(X_i = 1 | Q^\prime) \\
 &=  \frac{2}{2^n} \sum_{\substack{Q^\prime \in  \mathcal{P}_X \\ [Q^\prime,Q] = 0}}   \ty_l(Q^\prime U_{PQ}) = z^{PQ}_l \; .
 \end{align}
 The unbiasedness of  $\hat{z}_l^{P,i}$ hence follows from linearity of expectation.
Notice that the same $N^\prime$ measurement outcomes can be processed in different ways to get estimates $\hat{z}^{PQ}_l$ for all $l \in [2^n].$ These estimates can then be combined to get $\hat{z}_l^{P,i}.$

The number of shots required to guarantee a certain error in these estimates is given in Lemma \ref{lemma:est_z_I} and \ref{lemma:est_Pi}. The proofs use the Hoeffding's inequality and properties of subgaussian random variables.
In Algorithm \ref{alg:randomized} we describe how this estimates can be used to perform simultaneous tomography.

\begin{algorithm}[!ht]
\KwData{$\beta, \unorm{A}$}

 \tcp{{Refer Theorem \ref{thm:samples}} for the number of shots required for each step}
 \tcp{ \textcolor{purple}{Step 1. Find non-zero coefficients of the state}}
 $\mathcal{C} \leftarrow \{\}$ \tcp{Empty set}
\SetAlgoLined
Estimate $\{\hat{z}^I_k | k \in [2^n] \}$ \\
\For{$P \in \bc_R$}{

Estimate $\{\hat{z}^{P,i}_k | k,i \in [2^n] \}$ \\
    $s_P \leftarrow 0$\\
    \If{$\max_{k,i \in [2^n]}|\hat{z}^{P,i}_k - \hat{z}^{I}_k| \geq 1.005\beta \unorm{A}$}{   
    $\mathcal{C} \leftarrow \mathcal{C} \bigcup \{P\}$ 
    }
}
 \tcp{\textcolor{purple}{Step 2. Find $A$ up to gauge symmetry}}
 \textbf{choose} $R \in \mathcal{C}$ \\
 \For{$(k,i) \in [2^n]\times[2^n]$}{
 $A_{ki}' \leftarrow \hat{z}_k^{R,i}$}  
 \tcp{\textcolor{purple}{Step 3. Find other state coefficients up to a multiplicative constant}} 
Re-estimate $\{\hat{z}^{R,i}_k | k,i \in [2^n] \}$  with $\epsilon = \frac{\beta \unorm{A}}{4}$  \tcp{see Lemma \ref{lemma:est_Pi}} 
Re-estimate $\{\hat{z}^I_k | k \in [2^n] \}$  with $\epsilon = \frac{\beta \unorm{A}}{4}$  \tcp{see Lemma \ref{lemma:est_z_I}} 
 $i',l'\leftarrow \text{argmax}_{i,l}  |\hat{z}_l^{R,i} -\hat{z}_l^{I}|$ \\
\For{$P \in \mathcal{C}$}{

\tcp{In this step use  the $\hat{z}$ estimates that were compupted from the most amount of shots} 
$\frac{s_P}{s_R} \leftarrow \frac{\hat{z}^{P,i'}_{l'} -\hat{z}^{I}_{l'}}{ \hat{z}^{R,i'}_{l'} -\hat{z}^{I}_{l'}}$ 
}
\textbf{return} $\{ (P,\frac{s_P}{s_R}) | P \in  \mathcal{C}\},~~\Ap $
\caption{Finite-shot simultaneous tomography in computational basis}
\label{alg:randomized}
\end{algorithm}

\subsection{Proof of Theorem \ref{thm:samples}}

\begin{proof}
We will use randomized measurements to  estimate the $z$ values required for the algorithm.

 The definitions in \eqref{eq:EI_pauli_def} and \eqref{eq:EPi_def} show that the $z$ values for performing simultaneous tomography can be estimated by using simple Monte-Carlo estimates as defined in \eqref{eq:est_1} and \eqref{eq:est_3}.

 \noindent For instance, from the definition of $E_{I}$ in \eqref{eq:EI_pauli_def} we see that $z^I_k$ can be estimated by choosing a random unitary from $\pc_X$, applying it to the state and recording the noisy measurement outcomes. 
 The number of shots required to guarantee a certain accuracy in these estimates with high probability are given by Lemma \ref{lemma:est_z_I}  and Lemma \ref{lemma:est_Pi}. 
 
The covariance matrix for computational basis measurements is given by
$$ C_{ik} = \delta_{ik} - \frac{1}{2^n} \; .$$
From \eqref{eq:zPQ_def}, given any $P$, we have
\begin{equation}\label{eq:tempa}
   z^{P,i}_k  - z^{I}_k = s_P(A_{ki} - z^{I}_k) \; .
\end{equation}
The first step of Algorithm \ref{alg:general_decoding} is finding the non-zero state coefficients. To perform the analogous step in the  randomized setting we set a positive parameter $0 < \beta < \mnorm{\rho}$ and construct a set $\mathcal{C} \subset \pc$ such that for every $P \in C$ we can guarantee with high probability that $|s_P| > \beta$. Similarly, we can also show that with high probability, if $P \notin \mathcal{C}$ then $|s_P| < 1.01 \beta$. We show in Lemma \ref{lemma:sparsity} that such a set can be constructed using  $O( 8^n\frac{cn + \log(1/\delta)}{ \beta^2\unorm{A}^2})$ shots.

\noindent From this set $\mathcal{C}$ we can choose $R$ such that $s_R$ can be used  as the unknown gauge in the problem, which gives us the noise matrix upto $s_R$
 \begin{equation}
     A_{l\,i}^\prime(s_R) = z_l^{R,i} \; .
 \end{equation}
 Due to the simple form of the covariance matrix, the estimation error in the noise matrix is also given by Lemma \ref{lemma:est_Pi}.

\noindent For the final phase of the algorithm, we first have to find an element of the noise matrix that is sufficiently bounded away from its corresponding row average. To this end let us first estimate $|\hat{z}_l^{R,i} -\hat{z}_l^{I}|$ up to a max error of $\epsilon$ with high probability using Lemmas \ref{lemma:est_Pi} and \ref{lemma:est_z_I}. 
Now let $i^*, l^* = \text{argmax}_{i,l}  |z_l^{R,i} -z_l^{I}|.$  From \eqref{eq:tempa} $\text{max}_{i,l}  |z_l^{R,i} -z_l^{I}| = |s_R| \unorm{A}$. So  $|\hat{z}_{l^*}^{R,i^*} -\hat{z}_{l^*}^{I}|$ must be $\epsilon$ close to $|s_R| \unorm{A}$ w.h.p.  Let $i',l'=\text{argmax}_{i,l}  |\hat{z}_l^{R,i} -\hat{z}_l^{I}|$, then with high probability the following relations hold
\begin{align}
 |z_{l'}^{R,i'} -z_{l'}^{I}| &\geq |\hat{z}_{l'}^{R,i'} -\hat{z}_{l'}^{I}| - \epsilon \; , \\
 &\geq  |\hat{z}_{l^*}^{R,i^*} -\hat{z}_{l^*}^{I}| - \epsilon \; , \\
 &\geq  |s_R| \unorm{A}  -2\epsilon \; .
\end{align}
Substituting \eqref{eq:tempa} in the LHS of the above relation gives
%This, along with the fact that the max error in the estimates is at most $\epsilon$, will imply that $\max_{i,l}|\hat{z}_l^{Ri} -\hat{z}_l^{I}|$ will be at least $\epsilon$ close to $|s_R| \unorm{A}$.
%\noindent Now let $i',l'=\text{argmax}_{i,l}  |\hat{z}_l^{Ri} -\hat{z}_l^{I}|.$ So w.h.p, $|z_{l'}^{Ri'} -z_{l'}^{I}| = |s_R||A_{i'l'} -z_{l'}^I|$ will be $2 \epsilon$ close to $|s_R| \unorm{A}.$ This in turn implies that w.h.p
\begin{equation}
    |A_{i'l'} -z_{l'}^I| \geq \unorm{A} - \frac{ 2\epsilon}{|s_R|} \; .
\end{equation}

Since $|s_R| \in \mathcal{C}$, we have that $|s_R| \geq \beta$. So  choosing $\epsilon =  \frac{\unorm{A} \beta}{4}$ would give us $|A_{i'l'} -z_{l'}^I| \geq \unorm{A}/2$ w.h.p.. 

According to Lemma \ref{lemma:est_Pi}, the number of measurements required to find the indices $i'$ and $l'$ will be $O(2^n\frac{cn + \log(1/\delta)}{\beta^2 \unorm{A}^2})$.
Using these indices, for every $P \in \mathcal C$ we can estimate $\frac{s_P}{s_R}$ as
\begin{equation}
   \widehat{\frac{s_P}{s_R}}  = \frac{ |\hat{z}_{l'}^{P,i'} -\hat{z}_{l'}^{I}|}{ |\hat{z}_{l'}^{R,i'} -\hat{z}_{l'}^{I}|} \; .
\end{equation}
The key observation here is that the true values of both the numerator and the denominator in the above expression is greater than $\beta \unorm{A}/2$ because of the $i',l'$ indices we have chosen.
Now using this observation in Lemma \ref{lemma:multiplicative_err}, we show $O(2^n \frac{cn + \log(1/\delta)}{\epsilon^2 \beta^2 \unorm{A}^2})$ measurements are sufficient to get the above estimate to within $\epsilon$ in multiplicative error. 

Now we have to repeat this procedure for every $P \in \mathcal{C}$ with the same values of $i'$ and $l'$ to get estimates for the state coefficeints up to gauge.

From the standard union-bound argument we can show that a total of $O(2^n|\mathcal{C}| \frac{cn + \log(1/\delta)}{\epsilon^2 \beta^2 \unorm{A}^2})$ measurements are sufficient to ensure a maximum multiplicative error of $\epsilon$ with probability at least $1-\delta$.

\end{proof}

\subsection{Proof of Corollary \ref{corr:SP}}
\label{sec:corr_proof}
\begin{proof}
The argument here is similar to the proof of Theorem \ref{thm:samples}. To begin with, perform the first step of the randomized algorithm with a threshold $\beta = \epsilon \leq \mnorm{\rho}/2$ to construct a $\mathcal{C}$. We set $\hat{s}_P = 0$ for any $P \notin \mathcal{C}.$ This gives an estimate with additive error of $\epsilon$ for all $P \notin \mathcal{C}$. From the proof of Lemma \ref{lemma:sparsity}  we know that this construction requires the estimation of $|z^{P,i}_k - z_k^I|$ with error $\epsilon' = 0.005\epsilon \unorm{A}.$

Let
\begin{align}
&R,i',k' =   \underset{P \in \mathcal{C},i,k \in [2^n]}{\text{argmax}}  |\hat{z}^{P,i}_k - \hat{z}_k^I|\;, \\ 
&P^* = \underset{P \in \mathcal{C}}{\text{argmax}} |z^{P,i}_k - z_k^I|  =\underset{P \in \mathcal{C}}{\text{argmax}}{|s_P|}\;.
\end{align}
From the definition of $P^*$, we have $|z^{P^*,i}_k - z_{k}^I| =  \mnorm{\rho} |A_{ik} - z_i^I|\leq \mnorm{\rho} \unorm{A}.$

For any $i,k$ we have
\begin{align}
|s_R| |A_{ik} - z_i^I| &=  |z^{R,i}_k - z_k^I| \\
&\geq |\hat{z}^{R,i}_k - \hat{z}_k^I| - 0.005\epsilon \unorm{A} \\
&\geq |\hat{z}^{P^*,i}_k - \hat{z}_k^I| - 0.005\epsilon \unorm{A} \\
&\geq |z^{P^*,i}_k - z_k^I| - 0.01\epsilon \unorm{A} \\
&=  \mnorm{\rho} |A_{ik} - z_i^I|  -  0.01\epsilon \unorm{A}
\end{align}
Maximizing this inequality over all $i,k \in [2^n]$  we have
\begin{equation}
    |s_R| \geq \mnorm{\rho} -0.01 \epsilon \; .
\end{equation}
From this, with high probability for all $P \in \mathcal{C}$ we have
\begin{equation}
\left|\frac{s_P}{s_R}\right| \leq  \frac{\mnorm{\rho}}{\mnorm{\rho} -0.01\epsilon} \leq \frac{1}{0.995} \, .
\end{equation}

Here we have used the fact that $\epsilon \leq \mnorm{\rho}/2.$

Using the above fact in the multiplicative error bound in step 3 of Theorem \ref{thm:samples} gives

   $Pr\left( \max_{P \in \mathcal{C}}\; \left| \widehat{\frac{s_P}{s_R}} - \frac{s_P}{s_R} \right|  > \frac{\epsilon}{0.995}  \right) \leq \delta \;.$

   Notice that the error incurred in the last step is slightly higher than $\epsilon$. But this can be rectified by substituting $\epsilon$ with a slightly lower value ($0.995 \epsilon$) in the above procedure.

   The sample complexity for this procedure can be found by adding the sample complexities of the first and third step of Theorem \ref{thm:samples}.

\end{proof}

\subsection{Technical lemmas for randomized measurements}

\begin{lemma}[Estimating $z^I$ values]
\label{lemma:est_z_I}
For an $n-$qubit system, let  $N = O( \frac{cn + \log(1/\delta)}{\epsilon^2})$, for a constant $c < 10$. By post-processing $N$ randomized  noisy measurement outcomes obtained from applying a random operator in $\pc_X$ to the state $\rho$, we can find $\hat{z}^{I}_l$ such that
\begin{equation}
   Pr( \max_{l \in [2^n]} |\hat{z}^{I}_l - z^{I}_l| > \epsilon )  \leq \delta \; .
\end{equation}    
\end{lemma}
\begin{proof}

 Consider the estimate $\hat{z}_l^I$ computed as defined in \ref{eq:est_1}.
\noindent By Hoeffding's inequality, we can obtain a tail bound for these estimates
\begin{equation}
\label{eq:Hoeff_zI}
   Pr( |\hat{z}^{I}_l - z^{I}_l| > \epsilon)  \leq  2e^{-2N\epsilon^2} \; .
\end{equation}
Choosing $N =  \frac{c \log(2^n/\delta)}{\epsilon^2}$ for a constant $c$ gives us
\begin{equation}
   Pr( |\hat{z}^{I}_l - z^{I}_l| > \epsilon)  \leq  \frac{\delta}{2^n} \; .
\end{equation}
Let $\mathcal{A}_l$ be the event that  $|\hat{z}^{I}_l - z^{I}_l| \leq \epsilon$, then using the union bound
\begin{align}
   Pr( \max_{l \in [2^n]} |\hat{z}^{I}_l - z^{I}_l| \leq \epsilon ) &= Pr(\land^{2^n}_{l=1}   \mathcal{A}_l)\\
   &=  1 - Pr( \lor^{2^n}_{l=1}   \bar{\mathcal{A}}_l) \\
   &> 1 - \delta \; .
\end{align}
\end{proof}

\begin{lemma}[Estimating $z^{P,i}$ values]
\label{lemma:est_Pi}
For an $n-$qubit system, let $N = O( 2^n\frac{cn + \log(1/\delta)}{\epsilon^2})$, for a constant $c < 10$. Given a Pauli string $P \neq I$, by post-processing $N$ randomized  noisy measurement outcomes obtained from applying a random unitary from an efficiently characterizable subset of the Clifford group, we can find $\hat{z}^{Pi}_l$ such that
\begin{equation}
   Pr( \max_{l,i \in [2^n]} |\hat{z}^{Pi}_l - z^{Pi}_l| > \epsilon )  \leq \delta \; .
\end{equation}    
\end{lemma}
\begin{proof}

Consider the finite sample estimate defined in \eqref{eq:est_2} for $z^{PQ}_l$ using $N'$ measurement shots. Now, let $e_l^{PQ} =  \hat{z}^{PQ}_l - z^{PQ}_l$. By Hoeffding's inequality, we find that, $e^{PQ}_l$  is a sub-gaussian random variable \cite{vershynin2018high}:
\begin{definition}{\bf Subgaussian random variable}
A zero-mean random variable $X$ is Subgaussian with a variance proxy of $\sigma^2$ if
\begin{equation}
    Pr(|X| >t) \leq 2e^{-\frac{2t^2}{\sigma^2}} \; .
\end{equation}
We denote this by:
$  X  \sim \texttt{SubG}\left(\sigma^2 \right) \; .$
\end{definition}
The tail bound from Hoeffding's inequality gives us
\begin{equation}
   e_l^{PQ}  \sim \texttt{SubG}\left(\frac{1}{N^\prime} \right) \; .
\end{equation}

In total, estimating all $\hat{z}_l^{PQ}$ for all $Q \in \pc_Z \setminus {I}$ and $l \in 2^n$  requires $N = O(2^n N^\prime)$ independent measurement outcomes. 

\noindent Similarly, we can use $N''$ randomized measurements to estimate $z_l^{I}$ as in Lemma \ref{lemma:est_z_I}.
We define the error $e^{I}_l =  (1 -  \sum_{Q \in \pc_Z \setminus {I} } H_{iQ})(\hat{z}^{I}_l - z_l^{I}).$ From \eqref{eq:Hoeff_zI} we have
\begin{equation}
   e_l^{I}  \sim \texttt{SubG}\left(\frac{ (1 -  \sum_{Q \in \pc_Z \setminus {I} } H_{iQ})^2}{N''} \right) \; .
\end{equation}
Using these we can compute the following estimates $ \hat{z}_l^{P,i} = (1 -  \sum_{Q \in \pc_Z \setminus {I} } H_{iQ}) \hat{z}^{I}_l +   \sum_{Q \in \pc_Z \setminus {I} } H_{iQ} \hat{z}^{PQ}_l,$  for all $i,l \in [2^n].$ With the total number of measurements required being $N = N'' + 2^n N'$.

\noindent The error in this estimate can be computed as,
$e_l^{P,i} =  \hat{z}^{P,i}_l - z^{P,i}_l =  e^{I}_l +   \sum_{Q \in \pc_Z \setminus {I} } H_{iQ} e_l^{PQ}.$
This error is a linear combination of subgaussian random variables, which is also subgaussian by the following fact \cite{vershynin2018high}.
\begin{fact}{\bf (Theorem 2.6.3 in \cite{vershynin2018high})}
Let $\{X_i \sim \texttt{SubG}(\sigma^2_i)~~\forall  i \in [D] \}$ be independent, mean zero, subgaussian random variables, and $a = (a_1,\ldots, a_D ) \in \mathbb{R}^D$.
Then, for every $t \geq 0$, we have
$$ Pr( |\sum_i a_i X_i | > t ) \leq 2\exp \left( \frac{-2 t^2}{ \sigma^2 ||a||^2_2} \right),$$
where $\sigma^2 =  \max_i \sigma^2_i \; .$
\end{fact}
From this we have,
\begin{equation}\label{eq:subg1}
    e_l^{P,i} \sim \texttt{SubG} \left( \max \left(\frac{ (1 - \sum_{Q\neq I}H_{iQ})^2}{N''},  \frac{1}{N'} \right) (1 + \sum_{Q\neq I} H^2_{iQ}) \right) \; .
\end{equation}
Now, $H_{iQ}$ defined as $\frac{\braket{i|Q|i}}{2^{n/2}}$ is just the $n-$qubit Hadamard operator. From its unitarity we have $1 + \sum_{Q\neq I} H^2_{iQ} < 2 $.
The first row and column of $H$ (corresponding to $\ket{i=1}=\ket{0}\otimes \ldots \otimes \ket{0}$ and $Q = I$ respectively)  is a vector of all $1/\sqrt{2^n}.$ Again from unitarity of $H$ we have
\begin{align}
\sum_{Q \neq I} H_{1Q} = \sqrt{2^{n}} - \frac{1}{\sqrt{2^n}} \; ,\\
\sum_{Q \neq I} H_{iQ} =  - \frac{1}{\sqrt{2^n}},~ {i \neq 1} \; .
\end{align}
This gives in the worst case $(1 - \sum_{Q \neq I}H_{iQ})^2 < 2^{n}$.

\noindent Combining these upper-bounds in $\eqref{eq:subg1}$ we get
\begin{equation}
    e_l^{P,i} \sim \texttt{SubG} \left( 2 \max \left(\frac{2^n}{N''},  \frac{1}{N'} \right) \right) \; ,
\end{equation}
and if we choose $N'' = 2^n N'$, we get
    $e_l^{P,i} \sim \texttt{SubG} \left( \frac{2}{N'} \right) $, with the total number of measurements required being $N = 2^n N' + N'' = O(2^{n} N')$.
From the definition of Subgaussian random variables it follows that
$$Pr( |e_l^{P,i}| > \epsilon ) \leq  2\exp \left(- N^\prime  \epsilon^2 \right) \; .$$
Choosing then $N^\prime =  \frac{c \log(4^n/\delta)}{\epsilon^2}$ for a constant $c$ gives us
\begin{equation}
  Pr( |e_l^{P,i}| > \epsilon ) \leq  \frac{\delta}{4^n} \; .
\end{equation}
Using the same union bound argument used in Lemma \ref{lemma:est_z_I}, we can show that
\begin{equation}
  Pr( \max_{l,i \in [2^n]} |e_l^{P,i}| > \epsilon ) \leq  \delta \; ,
\end{equation}
and the total number of measurements required is $N = O(2^n N^\prime) =  O( 2^n\frac{cn + \log(1/\delta)}{\epsilon^2})$

\end{proof}

\begin{lemma}
\label{lemma:sparsity}
Given $0<\beta< \mnorm{\rho}/2$, we can construct $\mathcal{C} \subset \bc_R$ such that
\begin{enumerate}
\item For every $P \in \mathcal{C}$ we can guarantee with probability $1 -\delta$ that $|s_P| \geq \beta$.
\item For every $P \notin \mathcal{C}$  we can guarantee with probability $1-\delta$ that $|s_P| < 1.01\beta$.
\end{enumerate}
The construction of such a $\mathcal{C}$ requires a total of $O( 8^n\frac{cn + \log(1/\delta)}{ \beta^2\unorm{A}^2})$ randomized measurements.
\end{lemma}
\begin{proof}
    From the definition of $\unorm{A}$ we know that for every $P \in \pc$ 
\begin{equation}
    \label{eq:max_zP_zI}
\max_{i,k \in [2^n]}        |z^{P,i}_k - z^{I}_k| = |s_P|~\unorm{A} \; .
    \end{equation}
From Lemma \ref{lemma:est_z_I} and \ref{lemma:est_Pi}, for each $P$, we can estimate $|\hat{z}^{P,i}_k - \hat{z}^{I}_k|$ such that, with probability $1-\delta$, the maximum error in the LHS is at most $0.005 \beta\unorm{A}$ using $ O( 2^n\frac{cn + \log(1/\delta)}{ \beta^2 \unorm{A}^2})$ measurements. 

\noindent Once we compute these estimates for every $P \in \pc$ by using a total of  $ O( 8^n\frac{cn + \log(1/\delta)}{ \beta^2 \unorm{A}^2})$ measurements, we can define $\mathcal{C}$ such that
\begin{equation}
    \mathcal{C} = \{P~| \max_{k,i \in [2^n]}|\hat{z}^{P,i}_k - \hat{z}^{I}_k| \geq 1.005\beta \unorm{A}   \} \; .
\end{equation}
From the error in the estimates we can guarantee with high probability that for every $P \in \mathcal{C}$ we will have $ \max_{k,i \in [2^n]}|z^{P,i}_k - z^{I}_k| \geq \beta ~\unorm{A}$  and hence $|s_P| \geq \beta $.
Similarly, if $P \notin \mathcal{C}$, then $\max_{k,i \in [2^n]}|\hat{z}^{P,i}_k - \hat{z}^{I}_k| < 1.005\beta \unorm{A}$, which gives  $ \max_{k,i \in [2^n]} |z^{P,i}_k - z^{I}_k| < 1.01\beta ~\unorm{A}$   with high probability. This in turn implies that $|s_P| < 1.01 \beta$.

\end{proof}

\begin{lemma}
\label{lemma:multiplicative_err}
Suppose we have $P, R \in \mathcal{P} \setminus I$, such that $|s_P|, |s_R| > \alpha$. 
And we also know  $l,i \in [2^n],$  such that  $|A_{il} - z_l^{I}| >  \gamma$
. Then  we can estimate $\frac{s_P}{s_R}$ such that
\begin{equation}
   Pr( \left| \widehat{\frac{s_P}{s_R}} - \frac{s_P}{s_R} \right|  > \epsilon \left| \frac{s_P}{s_R} \right| ) \leq \delta \; ,
\end{equation}
using $N=O(2^n\frac{ cn + \log(1/\delta)}{\epsilon^2 \alpha^2 \gamma^2})$ randomized measurements.
\end{lemma}
\begin{proof}
We know that for indices satisfying the assumptions in the lemma
\begin{equation}
\frac{s_P}{s_R}= \frac{z_l^{P,i} - z^I_l}{ z_l^{R,i} - z^I_l} \; ,
\end{equation}
So an estimate for the ratio $s_P/s_R$ can be computed using ratios of the appropriate estimates of $z$ values obtained from measurements. From \eqref{eq:tempa} we also know that
\begin{equation}
    |z_l^{P,i} - z^I_l| \geq \alpha \gamma \; .
\end{equation}
From the assumption in the lemma, we know an $l, i$ for which this condition holds.
Now using $ N = O(2^n\frac{ cn +  \log(1/\delta)}{\epsilon^2 \alpha^2 \gamma^2 })$ randomized measurements we can find estimates for $z_l^{P,i} - z^I_l$ such that
\begin{equation}
Pr(| (\hat{z}^{P,i}_l - \hat{z}^I_l)- (z^{P,i}_l - z^I_l )| >  \frac{\epsilon~\alpha \gamma}{2}) \leq  \delta \; .
\end{equation}
Using Lemma \ref{lem:err_ratio}, we have w.h.p,
\begin{equation}
   \left| \frac{ \hat{z}^{P,i}_l - \hat{z}^I_l}{ \hat{z}^{R,i}_l - \hat{z}^I_l} - \frac{s_P}{s_R} \right| \leq \epsilon \left| \frac{s_P}{s_R} \right| \; .
\end{equation}
\end{proof}

\begin{lemma}[Error in  ratios] \label{lem:err_ratio}
Let  $\hat{x} = x + \epsilon_x $ and $ \hat{y} = y + \epsilon_y$  such that  $|x|, |y| \geq \alpha$ and  $|\epsilon_x|, |\epsilon_y|\leq\epsilon \leq  \frac{\alpha}{2}$. 
Then
\ca 
\begin{align}
 \frac{x}{y}(1 - c_1\frac{\epsilon}{\alpha}) \leq \frac{\hat{x}}{\hat{y}} \leq \frac{x}{y}(1 + c_2 \frac{\epsilon}{\alpha}),  ~~ \text{if}~~~~\frac{x}{y} > 0 \\
 \frac{x}{y}(1 +  c_2\frac{\epsilon}{\alpha}) \leq \frac{\hat{x}}{\hat{y}} \leq \frac{x}{y}(1 - c_1 \frac{\epsilon}{\alpha}),  ~~ \text{if}~~~~\frac{x}{y} < 0
\end{align}
\cb
\begin{equation}
   \left| \frac{\hat{x}}{\hat{y}} - \frac{x}{y} \right| \leq O \left( \frac{\epsilon}{\alpha} \right) \left|\frac{x}{y} \right| \; .
\end{equation}
\end{lemma}

\begin{proof}
\begin{align}
\frac{\hat{x}}{\hat{y}} = \frac{ x + \epsilon_x}{y  + \epsilon_y}  = \left(\frac{x}{y} \right) \frac{ 1 + \frac{\epsilon_x}{x}}{  1 + \frac{\epsilon_y}{y}} \; ,
\end{align}
so $  \frac{ 1 + \frac{\epsilon_x}{x}}{  1 + \frac{\epsilon_y}{y}}$ is the exact multiplicative error term that we have to bound. 

\noindent Since $| \epsilon_x/x | < \frac12$, we  have  $ \frac12 <  1 + \frac{\epsilon_x}{x} < \frac32$.
Also since  $| \epsilon_y/y | < \frac12$, we have, $ 1 - \epsilon_y/y \leq \frac{1}{ 1 + \epsilon_y/y} \leq 1 + 2 | \epsilon_y/y |.$

\paragraph{Lower bound on error}
\begin{align}
    \frac{ 1 + \frac{\epsilon_x}{x}}{  1 + \frac{\epsilon_y}{y}} &\geq     ( 1 + \frac{\epsilon_x}{x})(  1 - \frac{\epsilon_y}{y})  \\
     &=  1 + \frac{\epsilon_x}{x} - \frac{\epsilon_y}{y}( 1+   \frac{\epsilon_x }{x } ) \\
     &\geq 1 + \frac{\epsilon_x}{x} - c \frac{\epsilon_y}{y} \\
     & \geq 1 - \frac{\epsilon}{\alpha} - c \frac{\epsilon}{\alpha} \geq 1 - (1+c) \frac{\epsilon}{\alpha} \; .
\end{align}
Where  $c$ is either $\frac12$ or $\frac32$ depending on the sign of $\frac{\epsilon_y}{y}.$ 

\paragraph{Upper bound on error}
\begin{align}
    \frac{ 1 + \frac{\epsilon_x}{x}}{  1 + \frac{\epsilon_y}{y}} &\leq     ( 1 + \frac{\epsilon_x}{x})(  1 +  2|\frac{\epsilon_y}{y}|) \\
     &\leq  1 + \frac{\epsilon_x}{x} + 3 \left| \frac{\epsilon_y}{y} \right |  \\ 
     &\leq  1+ 4 \frac{\epsilon}{\alpha} \; .
\end{align}
Combining these two bounds on the multiplicative error gives us the claimed inequalities in the lemma.
  
\end{proof}

\section{Properties of measurement operators}
\label{app:m_op}

Recall that the overlap of the POVM on any traceless basis $\bc_L$ is defined as 
\begin{align}
    m_{kI} = \BK{M_{k}|\hat{I}}, \quad m_{kQ} = \BK{M_{k}|Q}, \ Q \in \bc_L.
\end{align}
Define the $D \times 4^n$ matrix $\Bm$ with elements given by $m_{kQ}$, and let $\Bm_{\setminus I}$ be the submatrix of $\Bm$ obtained by removing the column corresponding to $Q = I$. Then the following lemma holds.
\begin{lemma}   \label{lem:left_null_space}
If the POVM is linearly independent, then the matrix $\Bm$ has full row rank $D$. Furthermore, the matrix $\Bm_{\setminus I}$ has row rank $D-1$ and the only vector $v$ in the left null space that satisfies
\begin{align}
 v \Bm_{\setminus I} = 0,
\end{align}
is given by $v = \mathbbm{1}$ which is the vector of all ones.
\end{lemma}
\begin{proof}
The fact that the matrix $\Bm$ has full row rank $D$ follows directly from linear independence of the POVM. Consequently, the matrix $\Bm_{\setminus I}$ has row rank $D-1$ and must have a one dimensional left null space. By definition of POVM, $\sum_{k \in [D]} M_k = I$. Therefore, for all $Q \in \bc_L$ we get
\begin{align}
    \sum_{k \in [D]} m_{kQ} = \tr(Q \sum_{k \in [D]} M_k ) = \tr(Q I) = 0,
\end{align}
since $Q \in \bc_L$ is traceless.
\end{proof}

\section{Induced linear operator space of unitary operators}
Recall that for any subset $S \subseteq \uc$ of unitary operators on $n$ qubits, the induced linear operators space representing hybrid quantum-classical operations is given by \eqref{eq:linear_operator_space}
\begin{align}   
    \lc(\mathcal{S}) = \left\{\sum_{l}c_l \Phi(U_l) \mid U_l \in \mathcal{S},  \ \sum_{l} c_l = 1 \right\}.
\end{align}
These induced operator spaces $\lc(S)$ have many natural properties as given below.
\begin{lemma}[Closedness under linear combination] \label{lem:closed_linear}
    Let $S \subseteq \uc$ and let $\Phi_1, \ldots, \Phi_m \in \lc(S)$. Then for any  $c_1, \ldots, c_m$ such that $\sum_{i=1}^{m} c_i = 1$, we have $\sum_{i=1}^{m} c_i \Phi_i \in \lc(S)$. 
\end{lemma}
\begin{proof}
    Follows directly from the definition of $\lc(S).$ 
    $\Phi_i \in \lc(S)$ implies that  
    \begin{align}
    \Phi_i = \sum_{l=1}^{L_i} c_l^i \Phi(U_l^i), \quad \sum_{l=1}^{L_i} c_l^i = 1,\qquad  &i=1,2,
\end{align}
where all $U_{l}^i \in S$. Then $\sum_i  c_i \Phi_i = \sum_{i,l}  c_i c^i_l  \Phi(U^i_l) \in \lc(S) $
\end{proof}
\begin{lemma}[Closedness under composition] \label{lem:closed}
Let $S \subseteq \uc$ be such that for all $U_1, U_2 \in S$ we have that the unitary $U_1U_2 \in S$. Let $\Phi_1, \Phi_2 \in \lc(S)$. Then the composition of the these operators  $\Phi_1\Phi_2 \in \lc(S)$.
\end{lemma}
\begin{proof}[Proof of Lemma~\ref{lem:closed}]
By definition of $\lc(S)$ we have 
\begin{align}
    \Phi_i = \sum_{l=1}^{L_i} c_l^i \Phi(U_l^i), \quad \sum_{l=1}^{L_i} c_l^i = 1,\qquad  &i=1,2,
\end{align}
where all $U_{l}^i \in S$.
By using the composition of the two we get
\begin{align}
\Phi_1\Phi_2 &= \sum_{l_1,l_2} c_{l_1}^1c_{l_2}^2 \Phi(U_{l_1}^1)\Phi(U_{l_2}^2) \nonumber \\
&= \sum_{l_1,l_2} c_{l_1}^1c_{l_2}^2 \Phi(U_{l_1}^1U_{l_2}^2),
\end{align}
where we have used $\Phi(U_1)\Phi(U_2) = \Phi(U_1U_2)$ by definition of $\Phi(U)$ for a unitary $U$. The proof follows from the assertion that $U_{l_1}^1U_{l_2}^2 \in S$ and because the coefficients sum to one since $\sum_{l_1,l_2} c_{l_1}^1c_{l_2}^2 = \sum_{l_1} c_{l_1}^1 \sum_{l_2}c_{l_2}^2 = 1$.
\end{proof}
\begin{lemma}[Closedness under tensor product]   \label{lem:tensor_closed}
If $\Phi_1 \in \lc(S_1)$ and $\Phi_2 \in \lc(S_2)$ then $\Phi_1 \otimes \Phi_2 \in \lc(S_1 \otimes S_2)$.
\end{lemma}
\begin{proof}[Proof of Lemma~\ref{lem:tensor_closed}]
Similar to proof of Lemma ~\ref{lem:closed}. 
\begin{align}
\Phi_1 \otimes \Phi_2 &= \sum_{l_1,l_2} c_{l_1}^1c_{l_2}^2 \Phi(U_{l_1}^1) \otimes \Phi(U_{l_2}^2) \nonumber \\
&= \sum_{l_1,l_2} c_{l_1}^1c_{l_2}^2 \Phi(U_{l_1}^1 \otimes U_{l_2}^2).
\end{align}
The last implication holds because, by elementary linear algebra, the following relation holds 
\begin{equation}
    \Phi(A) \otimes \Phi(B) = \Phi(A \otimes B) \; .
\end{equation}
This is because for any operators $V$ and $W$ 
\begin{align*}
\Phi(A) \otimes \Phi(B) \Kket{V} \otimes \Kket{W} &= (AVA^\dagger) \otimes (BWB^\dagger) \\
&= (A \otimes B) (V \otimes W) (A \otimes B) ^\dagger \\ 
&=  \Phi(A \otimes B) \Kket{V} \otimes \Kket{W} \; . 
\end{align*}

\end{proof}

\section{Incorporating prior information}
\label{app:Priors}
\subsection{Block independent noise} \label{app:independent_noise}
\subsubsection{Proof of uniqueness}
We first show that under the assumption of block independence we can break the gauge degeneracy.
\begin{theorem}[Uniqueness for block-independent noise] \label{thm:independence_uniqueness}
\end{theorem}
Suppose that the POVM is described by $M_{kl} = M^1_k \otimes M^2_l$ where $k \in [D_1]$ and $l \in [D_2]$ with $D_1D_2 = D$.  Also let the qubit numbers for these two systems be $n_1$ and $n_2$, with $n_1 + n_2 = n$. For example, when the POVM is the computational basis, and the outcome of observing each qubit is binary valued as in \eqref{eq:bl_pauli}, this refers to a partitioning of $n$ qubits into two parts. Suppose that the noise acts independently on the two parts such that  $A = A^1\otimes A^2$ where $A^1$ acts on $M^1$ and $A^2$ acts on $M^2$. Let $A' = A'^1 \otimes A'^2$ be another such noise matrix and assume that  $A^1,A^2,A'^1,A'^2$ are not erasure type matrices discussed with Condition \ref{cond1}.
% in Eq.~\eqref{eq:cond1}.
Then if $A$ and $A'$ are related by the gauge equivalence \eqref{eq:gauge} we must have $A = A'$.
\begin{proof}%{Proof of Theorem~\ref{thm:independence_uniqueness}}
Since $A = A^1 \otimes A^2$ and $A' = A'^1 \otimes A'^2$, we will explicitly use double indices to represent each dimension.
From \eqref{eq:gauge} we have
\begin{align}   \label{eq:proof_temp1}
    A' = \alpha A + (1-\alpha) \text{diag}(d) \mathbbm{1} \; ,
\end{align}
where 
\revision{
\begin{align}
    d_{kl} = \frac{\sum_{\kp ,\lp} A_{kl, \kp \lp} m_{\kp \lp, I}}{2^{n/2}} = \frac{\sum_{\kp, \lp} A'_{kl, \kp\lp} m_{\kp \lp, I}}{2^{n/2}} \; .
\end{align}
Since $M_{kl} = M^1_K \otimes M^2_l$ we have $m_{kl, I} = m^1_{K,I} m^2_{l,I}$. Then
\begin{align}
     d_{kl} &= \frac{\sum_{\kp ,\lp} A_{kl, \kp \lp} m_{\kp \lp, I}}{2^{n/2}} \\
     &= \frac{\sum_{\kp} A^1_{k \kp} m^1_{k,I}} {2^{n_1/2}} \frac{\sum_{\lp} A^2_{l \lp} m^2_{l,I}} {2^{n_2/2}} 
     = d^1_k d^2_l \; ,
\end{align}
where for $i=1,2$ we define $d^i_k = \frac{\sum_{\kp} A^i_{k\kp} m^1_{k,I}} {2^{n_i/2}}$.}
Using the independence of the noise in \eqref{eq:proof_temp1}, we get
\begin{align} \label{eq:proof_temp2}
    A'^1 \otimes A'^2 = \alpha A^1 \otimes A^2 + (1-\alpha) \text{diag}(d^1) \mathbbm{1} \otimes\text{diag}(d^2) \mathbbm{1} \; .
\end{align}
Multiplying \eqref{eq:proof_temp2} on the left by $I \otimes \mathbf{1}_{D_1}^T$ and on the right by $I \otimes \mathbf{1}_{D_2}$, we get 
\begin{align}   \label{eq:A1_degen}
    A'^1 = \alpha A^1 +  (1-\alpha) \text{diag}(d^1) \mathbbm{1} \; .
\end{align}
A similar computation yields
\begin{align}   \label{eq:A2_degen}
    A'^2 = \alpha A^2 +  (1-\alpha) \text{diag}(d^2) \mathbbm{1} \; .
\end{align}
Using $A = A^1\otimes A^2$ and substituting \eqref{eq:A1_degen} and \eqref{eq:A2_degen} into \eqref{eq:proof_temp1} yields
\begin{align*}
& \alpha A^1\otimes A^2 + (1-\alpha) (\text{diag}(d^1) \mathbbm{1}_{D_1})\otimes (\text{diag}(d^2) \mathbbm{1}_{D_2}) = \nonumber \\
& (\alpha A^1 + (1-\alpha)  \text{diag}(d^1) \mathbbm{1}_{D_1}) (\alpha A^2 +  (1-\alpha) \text{diag}(d^2) \mathbbm{1}_{D_2}) \; .
\end{align*}
Rearranging the above gives
\begin{align*}
   \alpha(1-\alpha) (A^1  - \text{diag}(d^1) \mathbbm{1}_{D_1})(A^2 - \text{diag}(d^2) \mathbbm{1}_{D_2}) = 0 \; .
\end{align*}
Since none of $A, A^1, A^2$ can be the erasure channel by assertion, we must have $\alpha = 1$ and hence $A = A'$.

\end{proof}

\subsubsection{Algorithm to fix the gauge}
By Theorem~\ref{thm:independence_uniqueness}, no two noise matrices can be both block independent and be related by the gauge relation \eqref{eq:gauge}. Recall that Algorithm~\ref{alg:general_decoding} returns a candidate noise matrix $A'$ that is related to the true noise matrix $A$ by the relation
\begin{equation}\label{eq: gauge inv 2}
A = \alpha A' + (1-\alpha) \text{diag}(d') \mathbbm{1} \; .
\end{equation} 
Our goal is to find $\alpha$ such that the matrix $A$ decomposes as $A = A^1 \otimes A^2$.
Define the operations 
\begin{align}
&\text{T}_{1}[\cdot] = (\mathbf{1}_{D_1}^T \otimes I_{D_1}) (\cdot)(\mathbf{1}_{D_1} \otimes I_{D_2}) \; ,\\
&\text{T}_{2}[\cdot] = (I_{D_1} \otimes \mathbf{1}_{D_2}^T) (\cdot)(I_{D_1} \otimes \mathbf{1}_{D_2}) \; .
\end{align}
Since $A = A^1 \otimes A^2$ we have
\begin{align}   \label{eq:partial_tracing}
&\text{T}_{2}[A] = A^1 = \alpha \text{T}_{2}[A'] + (1-\alpha) \text{T}_{2}[ \text{diag}(d') \mathbbm{1}] \; , \\
&\text{T}_{S_1}[A] = A^2 = \alpha \text{T}_{1}[A'] + (1-\alpha) \text{T}_{1}[\text{diag}(d') \mathbbm{1}] \; .
\end{align}
Using \eqref{eq:partial_tracing} in  \eqref{eq: gauge inv 2}  we get
\begin{align}
 &( \alpha \text{T}_{2}[A'] + (1-\alpha) \text{T}_{2}[ \text{diag}(d') \mathbbm{1}])\\
 & \quad \otimes (\alpha \text{T}_{1}[A'] + (1-\alpha) \text{T}_{1}[\text{diag}(d') \mathbbm{1}]) \nonumber \\
 &= \alpha \tilde{A} + (1-\alpha) \text{diag}(d) \mathbbm{1} \; .
\end{align}
Since $A$ is not the erasure channel, we can assume $\alpha \neq 0$ and get 
\begin{align} \label{eq:finding_gaugE_{I}ndependence}
&\alpha ( \text{T}_{2}[A']\otimes  \text{T}_{1}[A'] -  \text{T}_{2}[A']\otimes \text{T}_{1}[ \text{diag}(d') \mathbbm{1}] \nonumber \\
&\qquad - \text{T}_{2}[ \text{diag}(d') \mathbbm{1}]\otimes \text{T}_{1}[A'] +  \text{T}_{2}[A']\otimes\text{T}_{1}[ \text{diag}(d') \mathbbm{1}] \nonumber \\
& =  \text{T}_{2}[A']\otimes\text{T}_{1}[ \text{diag}(d') \mathbbm{1}] + \text{T}_{2}[ \text{diag}(d') \mathbbm{1}]\otimes \text{T}_{2}[A'] \nonumber \\ 
&\qquad -  \text{diag}(d) \mathbbm{1}- A' \; ,
\end{align}
 The equation above is a matrix equality of the type $\alpha M^1 = M^2$, so we just need to find a matrix element $M^1_{ij} \neq 0$ such that $\alpha = M^2_{ij}/M^1_{ij}$. \\

\subsection{Linearly representable prior information} \label{app:linear_information}
Let $b_{S}^i, \ i=N_S +1,\ldots 4^n-1$ and $b_A^i, \ i=N_A+1,\ldots, D$ be any set of vectors that span the space orthogonal to $b_S^i, \ i=1\ldots, N_S$ and $b_A^i, \ i=1, \ldots, N_A$ respectively. Then simultaneous tomography can be performed by constructing canonical linear operators $E^{ij}$ which we describe below. Let $\Bm$ be the matrix of coefficients of the POVM given by $[\Bm]_{k,Q} = m_{kQ}$. By independence of the POVM, the matrix $\Bm$ has full rank. Therefore, we can construct vectors $\tilde{b}^j$ such that 
\begin{align}   \label{eq:tilde_b_construction}
    \sum_{Q \in \bc_R \cup I} \tilde{b}^j_Q m_{kQ} = b^j_{A,k} \; , ~~ k \in [D] \; , ~j = N_A+1, \ldots, D \; .
\end{align}
Let $E^{ij}$ be the linear operator such that its matrix representation using the bases $\bc_L, \bc_R$ is given by 
\begin{align}
    E^{ij}_{PQ} = b_{S,P}^i \tilde{b}^j_Q \; , \quad \forall P \in \bc_R \; , Q \in \bc_L \; .
\end{align}
Then, similar to \eqref{eq:z_def} and \eqref{eq:zPQ_def}, we can compute the following quantities using linear combinations of observations 
\begin{align}
    z_k^{ij} &= z_k^{I} + \sum_{\substack{\kp \in [D], \\ P \in \bc_R, Q \in \bc_L} } s_P E_{PQ}^{ij} A_{k \kp} m_{\kp Q} \nonumber \\
    & = z_k^{I} + (\sum_{P \in \bc_R} b_{S,P}^i s_P) ( \sum_{\kp \in [D]} A_{k \kp}\sum_{Q \in \bc_L} \tilde{b}^j_Q m_{\kp Q}) \nonumber \\
     & = z_k^{I} + (\sum_{P \in \bc_R} b_{S,P}^i s_P) ( \sum_{\kp \in [D]} A_{k \kp} (b^j_{A, \kp } - \tilde{b}_I^j m_{\kp I})) \nonumber \\
     &= z_k^{I} + (\sum_{P \in \bc_R} b_{S,P}^i s_P) ( \sum_{\kp \in [D]} A_{k \kp} b^j_{A, \kp } - \tilde{b}^j_I z_k^{I}) \; .
\end{align}
We also construct the linear operators required to fix the gauge denoted by $E^j$ and defined as
\begin{align}
    E^{j}_{PQ} = b_{S,P}^1 \tilde{b}^j_Q \; , \quad \forall P \in \bc_R \; , Q \in \bc_L \; .
\end{align}
and the corresponding computable quantity
\begin{align}
    z_k^{j} &= z_k^{I} + (\sum_{P \in \bc_R} b_{S,P}^1 s_P) ( \sum_{\kp \in [D]} A_{k \kp} b^j_{A, \kp } - \tilde{b}^j_I z_k^{I}) \nonumber \\
    &= z_k^{I}  + d_S^1 ( \sum_{\kp \in [D]} A_{k \kp} b^j_{A, \kp } - \tilde{b}^j_I z_k^{I}) \label{eq:gauge_linear} \; .
\end{align}
We will need to exploit the fact that $A$ is not the erasure channel to perform simultaneous tomography. This is given in the lemma below.
\begin{lemma} \label{lem:non_zero_linear_information}
    Assume that $A$ is not the erasure channel. Then there exists $k \in [D]$ for which there is a $j \in [D]$ such that  
    \begin{align}
        \sum_{\kp \in [D]} A_{k \kp} b^j_{A, \kp } - \tilde{b}^j_I z_k^{I} \neq 0 \; .
    \end{align}
\end{lemma}
\begin{proof}
Summing \eqref{eq:tilde_b_construction} for $k \in [D]$ we get that 
\begin{align}
    \tilde{b}^j_I = \sum_{k \in [D]} b^j_{A,k} \; .
\end{align}
Suppose that for a given $k \in [D]$ we have
\begin{align}
    \sum_{\kp} A_{k \kp} b_{A,\kp}^j = \tilde{b}^j_I z_{k}^{I} = (\sum_{\kp \in [D]} b_{A,\kp}^j) z_{k}^{I} \; , ~~~~ \forall j \in [D] \nonumber.
\end{align}
Since by construction the vectors $b_{A}^j$ form a complete basis of the $D$ dimensional space, we can invert the above relation to get $A_{k \kp} = z_{k}^{I}, ~~ \forall \kp \in [D]$, implying that $A$ is the erasure channel. Then proof follows from using the fact that $A$ is not the erasure channel.

\end{proof}

\begin{algorithm}[!ht]
Compute $z_k^{I}~~ \forall ~k \in [D]$ using \eqref{eq:z_def} \\

 \tcp{\textcolor{purple}{Step 1.}}
 \For{$j \in [M_A]$} {
 \For{$k \in  [D]$}{
 Compute $\sum_{\kp \in [D]} A_{k \kp} b^j_{A, \kp } - \tilde{b}^j_I z_k^{I} = \frac{z_k^{j} - z_k^{I}}{d_S^1}$.
 }}
 \tcp{\textcolor{purple}{Step 2. Find noise matrix $A$}}
\For{$k \in  [D]$}{
Solve the system of equations to obtain $A_{k \kp } ~~ \forall \kp \in [D]$: \\
$\sum_{\kp \in [D]} A_{k \kp} b^j_{A, \kp } - \tilde{b}^j_I z_k^{I} = \frac{z_k^{j} - z_k^{I}}{d_S^1}, ~~~~ j \in M_A$, \\
$\sum_{\kp \in [D]} A_{k \kp} c^j_{A, \kp } = d^j_A, ~~~~ j \in [N_A]$.
}
\tcp{\textcolor{purple}{Step 3. Find state $\rho$}}
Choose $k \in [D]$ and $j \in [D]$ as per Lemma~\ref{lem:non_zero_linear_information} such that $\sum_{\kp \in [D]}A_{k \kp} b^j_{A, \kp } - \tilde{b}^j_I z_k^{I} \neq 0$. \\
\For{$i \in [M_S]$}{
Compute $\sum_{P \in \bc_R} s_P b_{S,P}^i = \frac{z_k^{ij} - z_k^{I}}{\sum_{\kp \in [D]}A_{k \kp} b^j_{A, \kp } - \tilde{b}^j_I z_k^{I}}$
}
\textbf{Return} $\{s_P \mid P \in \bc_R,~~A\} $.

\caption{Simultaneous tomography with linear prior information}
\label{alg:linear_information}
\end{algorithm}

\subsection{Independent ancilla qubits} \label{app:ancilla}
By running Step~$1$ of Algorithm~\ref{alg:general_decoding} on the ancilla qubits, we can identify $P \in \bc_R^a$ and $i \in [D^a]$ such that 
\begin{align}
  z_k^{P i} - z_k^{I} =  s_{P}  \sum_{\kp} A_{k \kp}^a C_{\kp i}^a \neq 0 \; .
\end{align}
Construct the set of operators $\{E_Q \mid Q \in \bc_R^r\}$ given by
\begin{align}   \label{eq:eliminators_ancilla}
    &E_{Q} \Kket{P  \otimes Q} - \Bar{M}_i^{a} \otimes I \in (\bc_L^{a} \otimes I)^{\perp} \; , \nonumber \\
    &E_{Q } \Kket{P' \otimes Q'} \in (\bc_L^{a} \otimes I)^{\perp} ~~~~ \text{if} ~ P' \neq P ~\text{or} ~ Q' \neq Q \; .
\end{align}
The above operators can be constructed using linear combinations in \eqref{eq:linear_operators} since $P \otimes Q$ and $\Bar{M}_i^{anc} \otimes I$  are both traceless. Measuring the ancilla qubits independently is equivalent to tracing out the other measurements and effectively using the measurement operators $\{M_i^a \otimes I \mid i \in [D^{anc}]\}$. Using the operators in \eqref{eq:eliminators_ancilla}, we can compute the quantities for all $Q \in \bc_R^r$
\begin{align}
    z_k^{PQi} - z_k^{I} = s_P s_Q  \sum_{\kp} A_{k \kp}^a C_{\kp i}^a  \; .
\end{align}
We can then recover the state coefficients $\{s_Q \mid Q \in \bc_R^r\}$ as
\begin{align}
    s_Q = \frac{z_k^{PQi} - z_k^{I}}{ z_k^{P i} - z_k^{I} } \; .
\end{align}
We now compare to the setting in \cite{Lin_SPAM_2021} where they consider the hierarchical setting described in Sec.~\ref{sec: hierarchical dec}, with one ancillary qubit and orthogonal POVMs. In this case, the basis for the ancilla can be chosen to be $\{I, M_1^a-M_2^a\}$ where $M_1^a-M_2^a$ is traceless by definition of POVMs. Since the POVMs are orthogonal, any unitary operator $U_i$ such that 
\begin{align}
    U_i ((M_1^a - M_2^a) \otimes M_i^r) U_i^{\dagger} =  (M_1^a - M_2^a) \otimes I \; ,
\end{align}
will satisfy the conditions of the operator described in \eqref{eq:eliminators_ancilla}. The construction of this operator when the POVM is the computational $Z$-basis can be found in \cite{Lin_SPAM_2021}.

\subsection{Denoising the binary symmetric channel} \label{app:BSC}
For each $Q \in \pc^n$ the vector of measurement operator coefficients $\Bm_Q$ is an eigenvector of $A$. Since $\Bm_I = \mathbbm{1}$ we have $Am_I = \Bm_I$. For every other $Q \in \pc_Z^n$ the coefficients are given by $\Bm_Q = \otimes_{i=1}^n [1, \pm 1]$. Let $S_Q \subset [n]$ be the set of indices for which the component of $\Bm_Q$ is $[1,-1]$. Then we have $A \Bm_Q = \lambda_{S_Q} \Bm_Q$, where $\lambda_S = \prod_{i \in S_Q} (1-2p_i)$. Thus for any operator $\bphi$ the corresponding measurements are given by 
\begin{align}
    \by = \sum_{P,Q} s_P \bphi_{PQ} A \Bm_{Q}  = \sum_{P,Q} s_P \bphi_{PQ} \lambda_{S_Q} \Bm_{Q} \; .
\end{align}
This allows for the less expensive denoising in Algorithm~\ref{alg:binary_symmetric}. The generator gate set is given by
\begin{align}
    \gc_{BSC}  = \{I, SWAP(i,j), CNOT(i,j) \mid i,j \in [n]\}.
\end{align}

\begin{algorithm}
\SetAlgoLined
\tcp{Step 1:~Identifying non-zero state coefficients}
Initialize $\mathcal{S}_{nz} = \emptyset$\;
\For{$P \in \pc_Z^n \setminus I^n$}{
Compute $s_P \lambda_P  = [1\pm 1]^T \by$ \;
If $[1\pm 1]^T \by \neq 0$, update $\mathcal{S}_{nz} \leftarrow \mathcal{S}_{nz} \cup \{P\}$ \;
}
\tcp{Step 2:~Decode non-zero state coefficients}
\For{$P  = \bigotimes_{i=1}^n P_i \in \mathcal{S}_{nz}$}{
Identify $S = i \in [n] \mid P_i = Z$. \;
\uIf{$|S|=1$}
{Let $S = \{i\}$ and pick any $j \neq i$ \;
Obtain measurements $m_{P,G}$ for each gate $G \in \{I, CNOT(i,j), SWAP(i,j)\}$ \;
Decode $s_P \gets \frac{m_{P,I}m_{P,SWAP(i,j)}}{m_{P,CNOT(i,j)}}$ \;
} 
\Else{
Let $S = \{i_1 < i_2, \ldots < i_{|S|}\}$ \;
Obtain measurements $m_{P,G}$ for each gate $G \in \{I, CNOT(i_1,i_2), \ldots, CNOT(i_{|S|},1)\}$ \;
Compute $(1-2p_{j}) = \frac{m_{P,I}}{m_{P,CNOT(i,j)}}$\;
Decode $s_P \gets \frac{m_{P,I}}{\prod_{i \in S} (1-2p_i)}$\;
}
{\bf return} $\{s_P \mid P \in \pc_Z^n\}$\;
}
\caption{Denoising for binary symmetric output noise.}
\label{alg:binary_symmetric}
\end{algorithm}

\end{document}